\documentclass[12pt,a4paper]{amsart}
\usepackage[margin = 2cm]{geometry}
\usepackage{amsmath,amssymb,amsthm}
\usepackage{mathrsfs}
\usepackage{multirow}
\usepackage{diagbox}
\usepackage{graphicx}
\usepackage{color}
\usepackage{hyperref}

\newtheorem{theorem}{Theorem}

\newtheorem{lemma}{Lemma}
\newtheorem{proposition}{Proposition}
\newtheorem{assumption}{Assumption}

\theoremstyle{remark}

\newtheorem{remark}{Remark}

\def\Var{\mathrm{Var}}
\def\1{\mathbf{1}}
\def\ve{\varepsilon}
\def\hve{\hat{\ve}}
\def\hZ{\hat{Z}}
\def\shat{\hat{\sigma}}
\def\hKtheta{\widehat{K\theta}}
\def\hg{\hat{g}}
\def\hphig{\hat{\phi}_g}
\def\ivec{\mathbf{i}}
\def\idot{\ivec_{\bullet}}
\def\uvec{\mathbf{u}}

\def\xvec{\mathbf{x}}

\def\kvec{\mathbf{k}}
\def\Rhat{\hat{R}}
\def\hF{\hat{\mathbb{F}}}

\def\opn{o_P(n^{-1/2})}

\def\RR{\mathbb{R}}
\def\CC{\mathscr{C}}
\def\Zint{\mathbb{Z}}
\def\Lclass{\mathcal{L}}
\def\WW{\mathcal{W}}
\def\MM{\mathcal{M}}

\numberwithin{equation}{section}

\allowdisplaybreaks

\begin{document}

\title[GOF testing the distribution of errors]{Goodness-of-fit testing
the error distribution in multivariate indirect regression}

\begin{abstract}
We propose a goodness-of-fit test for the distribution of errors from
a multivariate indirect regression model. The test statistic is based
on the Khmaladze transformation of the empirical process of
standardized residuals. This goodness-of-fit test is consistent at the
root-$n$ rate of convergence, and the test can maintain power against
local alternatives converging to the null at a root-$n$ rate.
\end{abstract}

\author[J.\ Chown, N.\ Bissantz and H.\ Dette]{Justin Chown$^*$,
  Nicolai Bissantz and Holger Dette}

\thanks{$^*$ Correspondences may be addressed to Justin Chown
  (justin.chown@ruhr-uni-bochum.de). \\
  {\em Ruhr-Universit\"at Bochum,
    Fakult\"at f\"ur Mathematik, Lehrstuhl f\"ur Stochastik, 44780
    Bochum, DE}}

\maketitle

\bigskip



\section{Introduction}
\label{intro}

A common problem faced in applications is that one can only make
indirect observations of a physical process. Consequently, important
quantities of interest cannot be directly observed, but a suitable
image under some transformation is typically available. These
problems are  called inverse problems in the literature. Loosely
speaking, the goal is to recover a quantity $\theta$ (often a
function) from a distorted version of an image  $ K\theta $, where $K$
is some operator. Developing valid statistical inference procedures
for these inverse problems is desirable, and in recent years several
authors have worked on the construction of estimators, structural
tests, and (pointwise and uniform) confidence bands for the unknown
indirect regression function $\theta$ [see Mair and Ruymgaart (1996),
Cavalier and Tsybakov (2002), Johnstone et al.\ (2004),
Bissantz and Holzmann (2008), Cavalier (2008), Birke et al.\ (2010),
Johnstone and Paul (2014), Marteau and Math\'e (2014), and Proksch et
al.\ (2015) among many others]. In this paper we consider an indirect
regression model of the form
\begin{equation} \label{modeleq}
Y_j = \big[K\theta\big](X_j) + \ve_j,
 \qquad j = 1,\ldots,n,
\end{equation}
where $X_j$ is a predictor, $\ve_j$ is a random error and $K$ is a
convolution operator, which will be specified later (along with the
covariates $X_j$). Here $\theta$ is an unknown but square-integrable
smooth function. We study a unified approach to testing certain model
assumptions regarding the distribution function of the error $\ve_j$
in the indirect regression model \eqref{modeleq}.

Apart from specification of the operator $K$, many statistical
techniques used in applications for the estimation of $\theta$ depend
on the error distribution. For example, when recovering astronomical
images certain defects such as cosmic-ray hits are important to
identify and remove [Section $6$ of Adorf (1995)]. Here deviation
values between observations from pixels and an initial reconstruction
are calculated and compared with the standard deviation of the
noise. A large deviation indicates the presence of a possible
cosmic-ray hit, and observations from the affected pixels are
discarded (or replaced by imputed values) in subsequent iterative
reconstruction procedures that improve the quality of the final
reconstructed image. Determining an unrealistic deviation depends on
the structure of the noise distribution. More recently, Bertero et
al. (2009) review maximum likelihood methods for reconstruction of
distorted images, and, in their Section 5.2 on deconvolution using
sparse representation, these authors note the popularity of assuming
an additive Gaussian white noise model for transformed data. However,
it is not known in advance whether this transformation is appropriate
for a given image. If the transformation is inappropriate, then we can
expect the Gaussian white noise model to also be inappropriate. The
purpose of this paper is to help in answering some of these questions,
which could be considered as goodness-of-fit hypotheses of specified
error distributions.

Problems of this type have found considerable interest in direct
regression models (this is the case where $K$ is an identity operator
and only $\theta$ appears in \eqref{modeleq}) [see Darling (1955),
Sukhatme (1972) or Durbin (1973) for some early works or del Barrio et
al.\ (2000) and Khmaladze and Koul (2004) for more recent references].
However, to the best of our knowledge the important case of testing
distributional assumptions regarding the error structure of an
indirect regression model of the form \eqref{modeleq} has not been
considered so far. We address this problem by proposing a test, which
is based on the empirical distribution function of the standardized
residuals from an estimate of the regression function. The method is based
on a projection principle introduced in the seminal papers of
Khmaladze  (1981, 1988). This projection is also called the Khmaladze
transformation and it has been well-studied in the
literature. Exemplarily, we mention the work of Marzec and Marzec
(1997), Stute et al.\ (1998), Khmaladze and Koul (2004, 2009), Haywood
and Khmaladze (2008), Dette and Hetzler (2009), Koul and Song (2010),
M\"uller et al.\ (2012), and Can et al.\ (2015), who use the Khmaladze
transform to construct goodness-fit-tests for various problems. The
work which is most similar in spirit to our work is the paper of Koul
et al.\ (2018), who consider a similar problem in linear measurement
error models.

We prefer the projection approach because there is a common asymptotic
distribution describing the large sample behavior of the test
statistics (without unknown parameters to be estimated) and the
procedure can be easily adapted to handle different problems. To obtain a
better understanding of projection principles as they relate to
forming model checks, we direct the reader to consider the rather
elaborate work of Bickel et al.\ (2006), who introduce a general
framework for constructing tests of general semiparametric hypotheses
that can be tailored to focus substantial power on important
alternatives. These authors investigate a so-called {\em score
process} obtained by a projection principle. Unfortunately, the
resulting test statistics are generally not {\em asymptotically
distribution free}, i.e.\ the asymptotic distributions of these test
statistics generally depend on unknown parameters and inference using
them becomes more complicated. The Khmaladze transform is simpler to
specify and easily employed in regression problems, since test
statistics obtained from the transformation are asymptotically
distribution free with (asymptotic) quantiles immediately available.

The article is organized as follows. A brief discussion of Sobolev
spaces and their appearance in statistical deconvolution problems is
given in Section \ref{estimates}. In this section we further propose
an estimator of the indirect regression function and study its
statistical properties. The proposed test statistic is introduced in
Section \ref{norm_test}. Finally, Section \ref{simulations} concludes
the article with a numerical study of the proposed testing procedure
and an application. The technical
details and proofs of our results can be found in Section
\ref{appendix}.


\section{Estimating smooth indirect regressions}
\label{estimates}

Consider the model \eqref{modeleq} with the operator $K$ specifying
convolution between an unknown but smooth function $\theta$ and a
known distortion function $\psi$ that characterizes $K$, i.e.
\begin{equation} \label{modeleq1}
\big[K\theta\big](X_j) = \int_{\CC} \theta(\uvec)
 \psi(X_j - \uvec)\,d\uvec.
\end{equation}
Here the covariates $X_j$ are random and have support $\CC = [0,1]^m $
for some   $m \geq 1$. The model errors $\ve_1,\ldots,\ve_n$ are
assumed to be independent with mean zero and common distribution
function $F$ admitting a Lebesgue density function, which is denoted
by $f$ throughout this paper. We also assume that $\ve_1,\ldots,\ve_n$
are independent of the  i.i.d.\ covariates $X_1,\ldots,X_n$.

Throughout this article we will assume that the indirect regression
function $\theta$ from \eqref{modeleq} is periodic and smooth in the
sense that $\theta$ belongs to the subspace of {\em periodic, weakly
differentiable} functions from the class of square integrable
functions $\Lclass_2(\CC)$ with support $\CC$; see Chapter 5 of Evans
(2010) for definitions and additional discussion. For   $d \in
\mathbb{N}$  let $I(d)$ be the set of multi-indices $\ivec =
(i_1,\ldots,i_m)$ satisfying $\idot = i_1 + \dots + i_m \leq d$. To be
precise, we will call a function $q \in \Lclass_2(\CC)$ weakly
differentiable in $\Lclass_2(\CC)$ of order $d$ when there is a
collection of functions $\{q^{(\ivec)} \in \Lclass_2(\CC)\}_{\ivec \in
I(d)}$ such that
\begin{equation*}
\int_\CC\,q(\uvec)D^{\ivec}\varphi(\uvec)\,d\uvec
 = (-1)^{\idot}\int_\CC\,q^{(\ivec)}(\uvec)\varphi(\uvec)\,d\uvec,
 \qquad \ivec \in I(d),
\end{equation*}
for every infinitely differentiable function $\varphi$, with $\varphi$
and $D^{\ivec}\varphi$, $\ivec \in I(d)$, vanishing at the boundary of
$\CC$ and writing
\begin{equation*}
D^{\ivec}\varphi(\xvec) =
 \frac{\partial^{\idot}}{\partial x_1^{i_1}\dots\partial x_m^{i_m}}
 \varphi(\xvec),
 \qquad \xvec \in \CC.
\end{equation*}
The class of weakly differentiable functions from
$\Lclass_2(\CC)$ of order $d$ forms the Sobolev space
\begin{equation*}
\WW^{d,2}(\CC) = \bigg\{ q \in \Lclass_2(\CC) \,:\,
 q^{(\ivec)} \in \Lclass_2(\CC),
 ~\ivec \in I(d) \bigg\}.
\end{equation*}

The periodic Sobolev space $\WW_{\text{per}}^{d,2}$ are those
functions from $\WW^{d,2}$ that are periodic on $\CC$ and whose weak
derivatives are also periodic on $\CC$. An orthonormal basis for the
space $\Lclass_2(\CC)$ of square integrable functions is given by the
Fourier basis $\{e^{i2\pi\kvec\cdot\xvec}\,:\,\xvec \in \CC\}_{\kvec
\in \Zint^m}$. Here $\kvec \cdot \xvec = k_1x_1 + \dots + k_mx_m$ is
the common inner product between the vectors $\kvec = (k_1,\ldots,k_m)
\in \Zint^m$ and $\xvec = (x_1,\ldots,x_m) \in \CC$. It follows that
$\WW_{\text{per}}^{d,2}$ can be equivalently represented by
\begin{equation*}
\WW_{\text{per}}^{d,2} = \bigg\{ q \in \WW^{d,2}(\CC)\,:\,
 \sum_{\kvec \in \Zint^m} \big( 1 + \|\kvec\|^2 \big)^d
 |\varrho(\kvec)|^2
 < \infty \bigg\},
\end{equation*}
where $\|\cdot\|$ denotes the Euclidean norm and
\begin{equation*}
\varrho(\kvec) = \int_{\CC}\,q(\xvec)e^{-i2\pi\kvec\cdot\xvec}\,d\xvec,
 \quad \kvec \in \mathbb{Z}^m
\end{equation*}
are the Fourier coefficients of $q$ [see K\"uhn et al.\ (2014) for
further discussion]. The series in the equivalent representation of
$\WW_{\text{per}}^{d,2}$ motivates replacing the degree of weak
differentiability $d$ by a real-valued smoothness index $s >
0$. Throughout this article we work with the general indirect
regression model space $\MM(s)$ defined as
\begin{equation} \label{Mdef}
\MM(s) = \bigg\{ q \in \WW_{\text{per}}^{s,2}\,:\,
 \sum_{\kvec \in \Zint^m} \|\kvec\|^s||\varrho(\kvec)| < \infty
 \bigg\}.
\end{equation}

We will assume that $\theta \in \mathcal{M}(s_0)$, for some $s_0$
specified below, and that $\psi \in \Lclass_2(\CC)$ such that $\psi$ is
positive-valued and integrates to $1$ so that $K$ is a
convolution operator from $\Lclass_2(\CC)$ into $\Lclass_2(\CC)$. In
this case we can represent $K\theta$ in terms of a Fourier series
\begin{equation} \label{theta_expansion}
K\theta(\xvec) = \sum_{\kvec \in \Zint^m} R(\kvec)
 \exp\big(i2\pi\kvec\cdot\xvec\big)
 = \sum_{\kvec \in \Zint^m} \Psi(\kvec)\Theta(\kvec)
 \exp\big(i2\pi\kvec\cdot\xvec\big),
 \qquad \xvec \in \CC,
\end{equation}
where $\{R(\kvec)\}_{\kvec \in \Zint^m}$ and $\{\Theta(\kvec)\}_{\kvec
\in \Zint^m}$ are the Fourier coefficients of $K\theta$ and $\theta$,
respectively. In particular we have
\begin{equation} \label{fourierrelation}
\Theta (\kvec) = \frac{R(\kvec)}{\Psi (\kvec)}
\quad\text{for all}~ \kvec \in \Zint^m.
\end{equation}

Studying the indirect regression model
\eqref{modeleq} requires that we consider the ill-posedness of the
inverse problem. This phenomenon occurs because the ratio
$|R(\kvec)|/|\Psi(\kvec)|$ needs to be summable when $\theta \in
\MM(s)$. However, when estimated Fourier coefficients
$\{\Rhat(\kvec)\}_{\kvec \in \Zint^m}$ are used $|\Rhat(\kvec)|$ does
not asymptotically vanish (with increasing $\|\kvec\|$) due to the
stochastic noise from the errors $\ve_j$ in model
\eqref{modeleq}. Consequently, the ratio
$|\Rhat(\kvec)|/|\Psi(\kvec)|$ is not necessarily summable, and this
problem is therefore called ill-posed. We can see that the
coefficients $\{\Psi(\kvec)\}_{\kvec \in \Zint^m}$ determine the rate
at which the ratio $|\Rhat(\kvec)|/|\Psi(\kvec)|$ expands, and,
therefore, the ill--posedness of the inverse problem here is given by
the rate of decay in the coefficients $\{\Psi(\kvec)\}_{\kvec \in
\Zint^m}$ of the distortion function $\psi$. We will assume that the
inverse problem is mildly to moderately ill-posed in the sense of Fan
(1991):
\begin{assumption} \label{assumpPsi}
There are finite constants $b \geq 0$, $\gamma > 0$ and $0 \leq C_{\Psi}
< C_{\Psi}^*$ such that, for every $\|\kvec\| > \gamma$, the Fourier
coefficients $\{\Psi(\kvec)\}_{\kvec \in \Zint^m}$ of the function
$\psi$ in \eqref{modeleq1}  satisfy $C_{\Psi} \leq
\|\kvec\|^b|\Psi(\kvec)| < C_{\Psi}^*$.
\end{assumption}

Under Assumption \ref{assumpPsi}, whenever $\theta \in \MM(s_0)$, for
some $s_0 > 0$, it follows that $K\theta \in \MM(s_0 + b)$ from the
celebrated convolution theorem for the Fourier transformation. This
means that convolution of the indirect regression $\theta$ with the
distortion function $\psi$ {\em adds smoothness}, and the resulting
distorted regression function $K\theta$ is now smoother than $\theta$
by exactly the degree of ill-posedness $b$ of the inverse
problem. Note that Assumption \ref{assumpPsi} is milder than that of
Fan (1991) in the sense that we allow the degree of ill-posedness $b =
0$ and that the scaled Fourier coefficients can vanish. This covers
the case of direct regression models where $K$ is the identity
operator, that is $K \theta =\theta$. Further note that we do not have
to invert the operator $K$ in order to investigate properties of the
error distribution in the indirect regression model \eqref{modeleq}.

Several techniques have been developed in the literature to derive
series-type estimators (see, for example, Cavalier, 2008). A popular
regularization method to employ is the so-called {\em spectral
cut-off} method, where an indicator function is introduced in
\eqref{theta_expansion}. For example, the indicator function
$\1[\|c_n\kvec\| \leq 1]$ (for some sequence $\{c_n\}_{n \geq 1}$
converging to $0$) results in a biased version of $K\theta$:
\begin{equation*}
(K\theta)_n(\xvec) = \sum_{\kvec \in \Zint^m\,:\,\|\kvec\| \leq c_n^{-1}}
 R(\kvec) \exp\big(i2\pi\kvec\cdot\xvec),
\quad \xvec \in \CC.
\end{equation*}
The proposed estimator is obtained by replacing the coefficients
$\{R(\kvec)\}_{\kvec \in \Zint^m}$ with consistent estimators
$\{\Rhat(\kvec)\}_{\kvec \in \Zint^m}$, which gives
\begin{equation*}
 \sum_{\kvec \in \Zint^m\,:\,\|\kvec\| \leq c_n^{-1}}
 \Rhat(\kvec) \exp\big(i2\pi\kvec\cdot\xvec\big),
\quad \xvec \in \CC,
\end{equation*}
as an estimator of $(K\theta)_n$. The sequence of smoothing parameters
$\{c_n\}_{n\geq 1}$ is chosen such that $K\theta$ is consistently
estimated. We can generalize this approach as follows.

Following Politis and Romano (1999) we consider  a Fourier smoothing kernel
$\Lambda$, where $\Lambda$ is defined to be the Fourier transformation
of some smoothing kernel function, say  $L_{\Lambda}$. The resulting
estimate is then defined by
\begin{equation} \label{khat}
\hKtheta(\xvec) = \sum_{\kvec \in \Zint^m} \Lambda(c_n\kvec)
 \Rhat(\kvec) \exp\big(i2\pi\kvec\cdot\xvec\big),
 \quad \xvec \in \CC.
\end{equation}
Another useful observation that Politis and Romano (1999) make is
the function $\xvec \mapsto c_n^{-m}L_{\Lambda}(c_n^{-1}\xvec)$ has
Fourier coefficients $\{\Lambda(c_n\kvec)\}_{\kvec \in
\Zint^m}$. Throughout this paper we will choose $\Lambda$ as follows:
\begin{assumption} \label{assumpLambda}
The Fourier smoothing kernel $\Lambda$ satisfies $\Lambda(\kvec) = 1$,
for $\|\kvec\| \leq 1$, $|\Lambda(\kvec)| \leq 1$, for $\|\kvec\|
> 1$, and $\int_{\RR^m}\,\|\uvec\||\Lambda(\uvec)|\,d\uvec < \infty$.
\end{assumption}

The random covariates $X_1,\ldots,X_n$ from model \eqref{modeleq} are
assumed to be independent with distribution function $G$. For
simplicity we will assume that $G$ satisfies the following properties.
\begin{assumption} \label{assump_g}
Let the covariate distribution function $G$ admit a positive Lebesgue
density function $g \in \Lclass_2(\CC)$ satisfying $\inf_{\xvec \in
\CC} g(\xvec) > 0$, $\sup_{\xvec \in \CC} g(\xvec) < \infty$ and that
$g \in \MM(s)$ for some $s > 0$.
\end{assumption}
The boundedness assumptions taken for $g$ are common in nonparametric
regression because these conditions guarantee good performance of
nonparametric function estimators. The last condition ensures that the
density function $g$ satisfies similar smoothness properties as the
indirect regression function $\theta$, which allows us to use a
Fourier series technique to specify a good estimator of $g$ (see, for
example, Politis and Romano, 1999).

What remains is to define the estimates $\{\Rhat(\kvec)\}_{\kvec \in
\Zint^m}$ of the Fourier coefficients $\{R(\kvec)\}_{\kvec \in
\Zint^m}$ required in the definition \eqref{khat}. Observing the
representation
\begin{equation*}
R(\kvec) = \int_{\CC}\,\big[ K\theta \big](\xvec)
 e^{-i 2\pi \kvec \cdot \xvec}\,d\xvec
 = E\bigg[ \frac{Y}{g(X)} e^{-i 2\pi \kvec \cdot X} \bigg],
 \quad \kvec \in \Zint^m,
\end{equation*}
the covariate density function $g$ must be estimated.
For this purpose we the expand the density function $g$ into its
Fourier series using the coefficients $\{\phi_g(\kvec)\}_{\kvec \in
\Zint^m}$, with $\phi_g(\kvec) = E[\exp( -i 2\pi \kvec \cdot X
)]$. Estimators of these coefficients are given by
\begin{equation*}
\hphig(\kvec) = \frac1n \sum_{j = 1}^n
 e^{-i 2\pi \kvec \cdot X_j},
 \quad \kvec \in \Zint^m.
\end{equation*}
From these estimators we then obtain an estimator $\hg$ of the unknown
covariate density function $g$, that is
\begin{equation} \label{cg_def}
\hg(\xvec) = \frac1n \sum_{j = 1}^n
 W_{c_n}\big( \xvec - X_j \big),
 \quad \xvec \in \CC,
\end{equation}
with smoothing weights
\begin{equation} \label{weights}
W_{c_n}\big( \xvec - X_j \big) =
 \sum_{\kvec \in \Zint^m} \Lambda(c_n\kvec)
 \exp\Big\{ i 2\pi \kvec \cdot \big( \xvec - X_j \big) \Big\}.
\end{equation}
Here (as before) the choice of $\Lambda$ defines the form of the
smoothing weights $W_{c_n}$. The sequence $\{c_n\}_{n \geq 1}$ of
smoothing parameters is specified later.

We now propose to estimate the Fourier coefficients
$\{R(\kvec)\}_{\kvec \in \Zint^m}$ of the distorted regression
function $K\theta$ by
\begin{equation*}
\Rhat(\kvec) = \frac1n \sum_{j = 1}^n
 \frac{Y_j}{\hg(X_j)} e^{ -i 2\pi \kvec \cdot X_j },
 \quad \kvec \in \Zint^m,
\end{equation*}
where the density estimator $\hg$ is specified in \eqref{cg_def}. This
gives for the nonparametric Fourier series estimator in \eqref{khat}
the representation
\begin{equation} \label{htheta_random_x}
\hKtheta(\xvec) = \sum_{\kvec \in \Zint^m} \Lambda(c_n\kvec) \Rhat(\kvec)
 e^{i 2\pi \kvec \cdot \xvec }
= \frac1n \sum_{j = 1}^n
 \frac{Y_j}{\hg(X_j)} W_{c_n}\big( \xvec - X_j \big),
\quad \xvec \in \CC,
\end{equation}
where the smoothing weights $W_{c_n}$ are defined in
\eqref{weights}.

The results of Lemma \ref{lem_Rhat} in Section \ref{appendix} show
that the consistency of the estimated Fourier coefficients
$\{\Rhat(\kvec)\}_{\kvec \in \Zint^m}$ is heavily dependent on the
consistency of the covariate density estimator $\hg$. This fact
motivates our choice of smoothing parameters as
\begin{equation} \label{bw}
c_n = O\big(n^{-1/(2s_0 + 2b + 3m)} \log^{1/(2s_0 + 2b + 3m)}(n)\big)
\end{equation}
and requiring that the covariate density function $g$ has a smoothness
index $s = s_0 + b + m$ in Assumption \ref{assump_g}, where $s_0$
is the smoothness index of the function class $\MM(s_0)$ to which
$\theta$ belongs, $b$ is the degree of ill-posedness of the inverse
problem and $m$ is the dimension of the covariates. Our first result
result establishes the uniform consistency of the estimator $\hKtheta$
in \eqref{khat} and a further technical metric space inclusion
property that is useful for working with residual-based empirical
processes.

\begin{theorem} \label{thm_htheta_random}
Let $\theta \in \MM(s_0)$ for some $s_0 > 0$ and let Assumption
\ref{assumpPsi} hold for some degree of
ill-posedness $b \geq 0$. Let Assumption \ref{assumpLambda} hold for a
Fourier smoothing kernel $\Lambda$ that satisfies
$\int_{\RR^m}\,\|\uvec\|^{\max\{s_0+b,1\}}|\Lambda(\uvec)|\,d\uvec <
\infty$. Further let Assumption \ref{assump_g} hold for $s =
s_0 + b + m$ and assume that the errors $\ve_1,\ldots,\ve_n$ have a
finite absolute moment of order $\kappa > 2$. Choose the smoothing
parameter $c_n$ as in \eqref{bw}. Then
\begin{equation*}
\sup_{\xvec \in \CC} \Big| \hKtheta(\xvec) - K\theta(\xvec) \Big|
 = O\big( n^{-(s_0 + b)/(2s_0 + 2b + 3m)}
 \log^{(s_0 + b)/(2s_0 + 2b + 3m)}(n) \big),
 \quad\text{a.s.,}
\end{equation*}
and
\begin{equation*}
\hKtheta - K\theta \in \MM_1(s_0 + b),
 \quad\text{a.s.,}
\end{equation*}
where $\MM_1(s_0 + b)$ is the unit ball of the metric space
$(\MM(s_0 + b),\,\|\cdot\|_{\infty})$.
\end{theorem}


\section{Goodness-of-fit testing the error distribution}
\label{norm_test}

In this section we consider the problem of goodness-of-fit testing of
a location-scale distribution of the errors in the indirect regression
model \eqref{modeleq} with convolution operator
\eqref{modeleq1}. Here the location parameter is the mean of the
errors and equal to zero, but the scale parameter is unknown. The null
hypothesis is given by
\begin{equation} \label{hyp_null}
H_0 \,:\, \exists \sigma > 0 \,:\, f(t) =
 \frac1{\sigma}f_*\bigg(\frac{t}{\sigma}\bigg),
 \quad t \in \RR,
\end{equation}
where $f_*$ is a specified density function of the standardized error
distribution and $\sigma$ is the unknown scale parameter. To simplify
notation we write $f_\sigma$ for the density function of the {\em
standardized errors} $Z_j = \ve_j / \sigma$ ($j=1,\ldots,n$) and
$F_\sigma(t) = \int_{-\infty}^t\,f_\sigma(y)\,dy $ ($t \in \RR$) for
the corresponding distribution function. With this notation the null
hypothesis in \eqref{hyp_null} becomes $H_0\,:\,f_\sigma = f_*$ for
some $\sigma > 0$. Equivalently, we can write $H_0\,:\,F_\sigma = F_*$
for some $\sigma >0$ by writing $F_*(t) =
\int_{-\infty}^t\,f_*(y)\,dy$ ($t \in \RR$) for the error distribution
function specified by the null hypothesis.

Following  M\"uller et al.\ (2012), who consider a similar problem in
the direct case, we propose to use the standardized residuals
\begin{equation*}
\hZ_j = \frac{\hve_j}{\shat},  \quad j = 1,\ldots,n,
\end{equation*}
to form a suitable test statistic, where $\hve_j = Y_j -
\hKtheta(X_j)$ ($j = 1,\ldots,n$) are the residuals in the indirect
regression model \eqref{modeleq} obtained for the estimate
\eqref{htheta_random_x} and
\begin{equation*}
\shat = \bigg\{ \frac1n \sum_{j = 1}^n \hve_j^2 \bigg\}^{1/2}
\end{equation*}
is a consistent estimator of the scale parameter $\sigma$. A
nonparametric estimator of $F_*$ is given by the empirical
distribution function of these standardized residuals,
\begin{equation*}
\hF(t) = \frac1n \sum_{j = 1}^n \1\big[ \hZ_j \leq t \big],
 \quad t \in \RR.
\end{equation*}

The null hypothesis $H_0$ is then rejected if a given metric between
the estimated standardized distribution function $\hF$ and
$F_*$ is large enough. A popular metric in the literature is the
supremum metric, and this leads to the Kolmogorov-Smirnov test
statistic:
\begin{equation*}
\sup_{t \in \RR} \Big| \hF (t) - F_*(t) \Big|.
\end{equation*}
Critical values for the Kolmogorov-Smirnov test statistic are then
determined from asymptotic theory, but these can be difficult to work
with in practice because they depend on $F_*$. To avoid this problem,
we will work with a different test statistic.

Our proposed test statistic will crucially depend on the estimator
$\hF$ satisfying an asymptotic expansion, which is given in the
following result.
\begin{theorem} \label{thm_hF_aslin_H0f}
Let the assumptions of Theorem \ref{thm_htheta_random} hold, with $s_0
+ b > 3m/2$ and assume that the Fourier smoothing kernel $\Lambda$ is
radially symmetric. Let $F_*$ have a finite absolute moment of
order $4$ or larger and  a bounded Lebesgue density
$f_*$ that is (uniformly) H\"older continuous with exponent $3m/(2s_0
+ 2b) < \gamma \leq 1$. Finally, the function $t \mapsto tf_*(t)$ is
assumed to be uniformly continuous and bounded. Then under the null
hypothesis \eqref{hyp_null}
\begin{equation*}
\hF(t) - F_*(t) = \frac1n \sum_{j = 1}^n \bigg\{
 \1[ Z_j \leq t ] - F_*(t)
 + f_*(t) \bigg( Z_j + t\frac{Z_j^2 - 1}{2} \bigg) \bigg\}
 + D_n(t), \quad t \in \RR,
\end{equation*}
with $\sup_{t \in \RR}|D_n(t)| = \opn$.
\end{theorem}

\begin{remark} \label{rem_simulated_critvals}
A direct consequence of Theorem \ref{thm_hF_aslin_H0f} is that,
under the null hypothesis \eqref{hyp_null}, the stochastic process $\{
\sqrt{n}( \hF(t) - F_*(t) ) \}_{t \in \RR}$ weakly converges in the space
$\ell^\infty([-\infty,\,\infty])$ to a Gaussian process, which is also
the weak limit of the stochastic process
\begin{equation*}
\bigg\{ \frac{1}{\sqrt{n}}
 \sum_{j = 1}^n \bigg\{
 \1[ Z_j \leq t ] - F_*(t)
 + f_*(t)\bigg(  Z_j + t\frac{Z_j^2 - 1}{2} \bigg)
 \bigg\} \bigg\}_{t \in \RR}.
\end{equation*}
This limit distribution can be easily simulated. However, it is
clearly not distribution free because it depends on $F_*$ and $f_*$
specified in the null hypothesis.
\end{remark}

In order to obtain a test statistic whose critical values are
independent from the distribution specified in the null hypothesis, we
use a particular projection of the residual-based empirical process by
viewing this quantity as an (approximate) semimartingale with respect
to its natural filtration. The projection is given by the Doob-Meyer
decomposition of this semimartingale (see page 1012 of Khmaladze and
Koul, 2004). For this purpose we will assume that $F_*$ has finite
Fisher information for location and scale, i.e.\
\begin{equation} \label{F_Fisher}
\int_{-\infty}^{\infty} \big( 1 + t^2 \big)
 \bigg(\frac{f_*'(t)}{f_*(t)} \bigg)^2\, F_*(dt)
 < \infty,
\end{equation}
writing $f_*'$ for the derivative of the Lebesgue density $f_*$.

The Khmaladze transformation produces a standard limiting
distribution: a standard Brownian motion on $[0,\,1]$, and as a
consequence we can construct test statistics which are asymptotically
distribution free, i.e.\ the corresponding critical values do not
depend on $F_*$ specified by the null hypothesis.

To be precise, note that $F_*$ characteristically has mean
zero and variance equal to one. In order to introduce our test
statistic we define the augmented score function
\begin{equation*}
h(t) = (1, -f_*'(t)/f_*(t), -(tf_*(t))'/f_*(t))^T
\end{equation*}
and the incomplete information matrix
\begin{equation} \label{GammaMat}
\Gamma(t) = \int_t^\infty\,h(u)h(u)^T\,F_*(du),
 \quad t \in \RR.
\end{equation}
Following Khmaladze and Koul (2009) the transformed empirical process
of standardized residuals is given by
\begin{equation*}
\hat{\xi}_0(t) = n^{1/2}\bigg\{ \hF(t) - \int_{-\infty}^t\,
 h^T(y)\Gamma^{-1}(y)\int_y^\infty\,h(z)\hF(dz)\,F_*(dy)
 \bigg\},
\quad -\infty < t \leq t_0,
\end{equation*}
for some $t_0 < \infty$. We can rewrite $\hat{\xi}_0$ in a more
computationally friendly form, i.e.\
\begin{equation*}
\hat{\xi}_0(t) = n^{1/2} \bigg\{ \hF(t) - \frac1n \sum_{j = 1}^n
 \mathscr{G}_0\big(t \wedge \hZ_j \big) h\big( \hZ_j \big) \bigg\},
 \quad -\infty < t \leq t_0,
\end{equation*}
where
\begin{equation*}
\mathscr{G}_0(t) = \int_{-\infty}^t\,
 h^T(y)\Gamma^{-1}(y)\,F_*(dy),
 \quad -\infty < t \leq t_0.
\end{equation*}
Under the null hypothesis \eqref{hyp_null} $\hat{\xi}_0$ weakly
converges in the space $\ell^\infty([-\infty,\,t_0])$ to
$\mathscr{B}(F_*)$, writing $\mathscr{B}$ for the standard Brownian
motion.

In general, the incomplete information matrix $\Gamma$ does not have a
simple form, and $\Gamma(t_0)$ degenerates  as $t_0
\to \infty$. To avoid this degeneracy issue we proceed as in Stute et
al.\ (1998), who recommend using the $99\%$ quantile from the
empirical distribution function $\hF$ for $t_0$, i.e.\ $t_0 =
\hF^{-1}(0.99)$ writing $\hF^{-1}$ for the sample quantile function
associated with $\hF$. We propose to base a goodness-of-fit test for
the hypothesis \eqref{hyp_null} on the supremum metric between
$\hat{\xi}_0 / (\hF(t_0))^{1/2}$ and the constant $0$:
\begin{equation} \label{test_stat_H0f}
T_0 = \sup_{-\infty < t \leq t_0} \bigg|
 \frac{\hat{\xi}_0(t)}{(\hF(t_0))^{1/2}} \bigg|
 = \sup_{-\infty < t \leq t_0} \bigg|
 \frac{\hat{\xi}_0(t)}{0.995} \bigg|.
\end{equation}
The test statistic $T_0$ has an asymptotic distribution given by
$\sup_{0 \leq s \leq 1}|\mathscr{B}(s)|$ under the null hypothesis
\eqref{hyp_null}.

Our proposed goodness-of-fit test for the null hypothesis
\eqref{hyp_null} is then defined by
\begin{equation} \label{test}
\text{Reject $H_0$ when $T_0 > q_\alpha$,}
\end{equation}
where $q_\alpha$ is the upper $\alpha$-quantile of the
distribution of $\sup_{0 \leq s \leq 1} |\mathscr{B}(s)|$. The value
of $q_\alpha$ may be obtained from formula (7) on page 34 of Shorack
and Wellner (1986), i.e.
\begin{equation*}
P\bigg( \sup_{0 \leq s \leq 1} \big| \mathscr{B}(s) \big| > q_\alpha \bigg)
 = 1 - \frac{4}{\pi} \sum_{k = 0}^\infty \frac{(-1)^k}{2k + 1}
 \exp\bigg(-\frac{(2k + 1)^2\pi^2}{8q_\alpha^2}\bigg),
\quad \alpha < 1.
\end{equation*}
For a $5\%$-level test, $\alpha = 0.05$ and $q_{0.05}$ is
approximately $2.2414$.


\section{Finite sample properties}
\label{simulations}

We conclude the article with a numerical study of the previous
results with two examples and an application of the proposed
test. Throughout this section we consider a goodness-of-fit
test for normally distributed errors in the indirect regression model
\eqref{modeleq}, i.e.\
\begin{equation*}
H_0\,:\,F_\sigma = \Phi
 \quad \text{for some $\sigma > 0$}.
\end{equation*}
Note that in this case a straightforward calculation
shows that the augmented score function $h$ and the incomplete
information matrix $\Gamma$ from \eqref{GammaMat} become particularly
simple, that is $h(t) = (1,\,t,\,t^2-1)^T$ and
\begin{equation*}
\Gamma(t) = \begin{pmatrix}
1 - \Phi(t) & \phi(t) & t\phi(t) \\
\phi(t) & 1 - \Phi(t) + t\phi(t) & (t^2 + 1)\phi(t) \\
t\phi(t) & (t^2 + 1)\phi(t) & 2(1 - \Phi(t)) + (t^3 + t)\phi(t)
\end{pmatrix},
\quad t \in \RR,
\end{equation*}
writing $\Phi$ and $\phi$ for the respective distribution and density
functions of the standard normal distribution.

\subsection{Simulation study}
\label{sim_study}
In the first example we generate independent bivariate covariates $X_j
= (X_{1,j},\,X_{2,j})^T$ with independent and identically distributed
components $X_{1,j}$ and $X_{2,j}$ ($j=1,\ldots,n$) as follows. The
common distribution of $X_{1,j}$ and $X_{2,j}$ is characterized by the
density function $g(x_1,\,x_2) = g_1(x_1)g_1(x_2)$ ($(x_1,\,x_2)^T \in
[0,\,1]^2$), which is depicted in the left panel of Figure
\ref{persp_plots}, where
\begin{equation*}
g_1(x) = 1 - \frac{\sqrt{2}}{4} \cos( 2 \pi x )
 - \frac{\sqrt{2}}{8} \cos( 4 \pi x ),
\quad x \in [0,\,1].
\end{equation*}
One can easily verify that $g$ is a probability density function and
satisfies the requirements of Assumption \ref{assump_g} for any $s >
0$. The random sample of covariates $X_1,\ldots,X_n$ is then generated
from the distribution characterized by the non-trivial density
function $g$ using a standard probability integral transform
approach. In the second example we use independently, uniformly
distributed covariates in the unit square $[0,\,1]^2$.

\begin{figure}
\centering
\begin{minipage}{0.3\textwidth}
\includegraphics[width=0.9\textwidth]{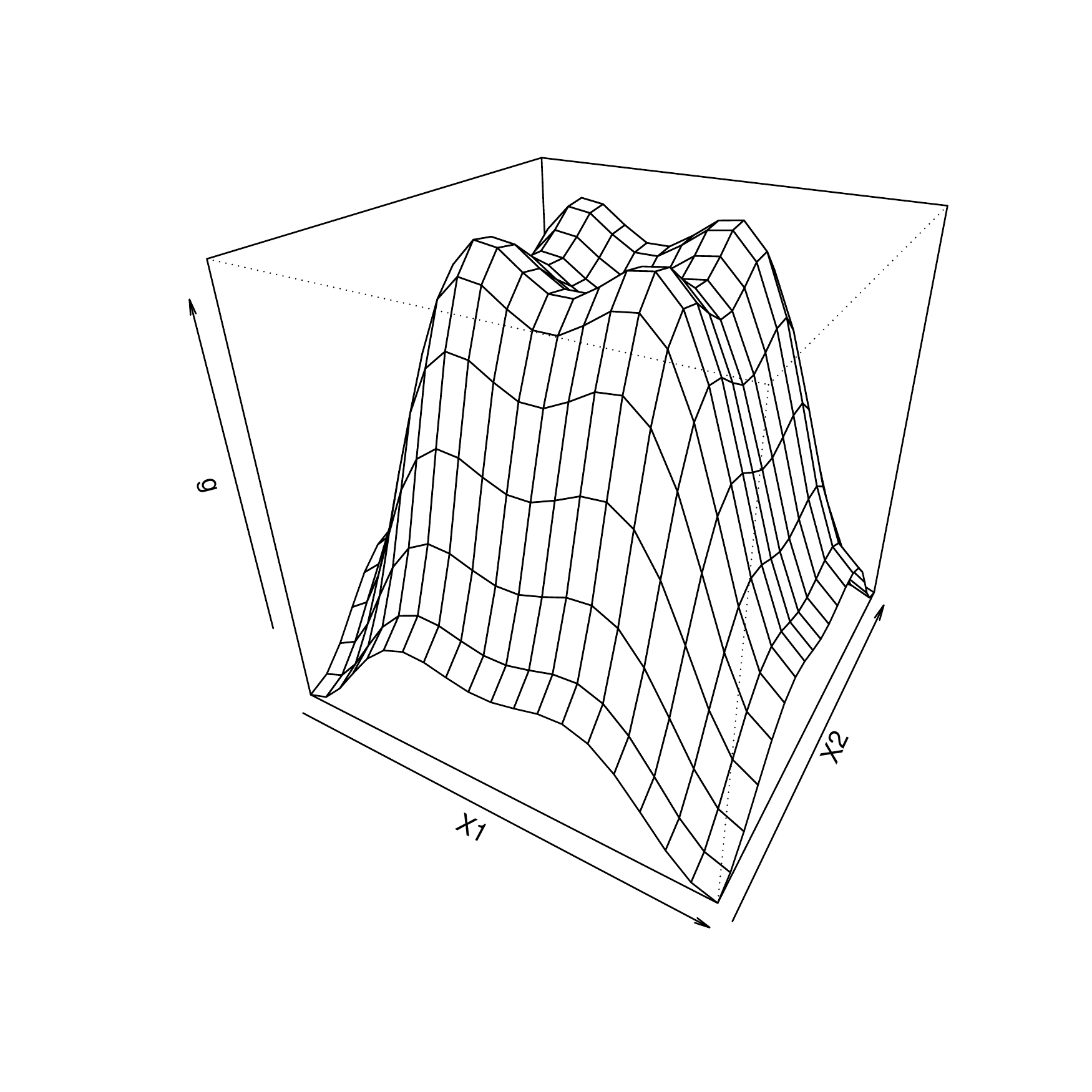}
\begin{center}{\it (a)}\end{center}
\end{minipage}
\begin{minipage}{0.3\textwidth}
\includegraphics[width=0.9\textwidth]{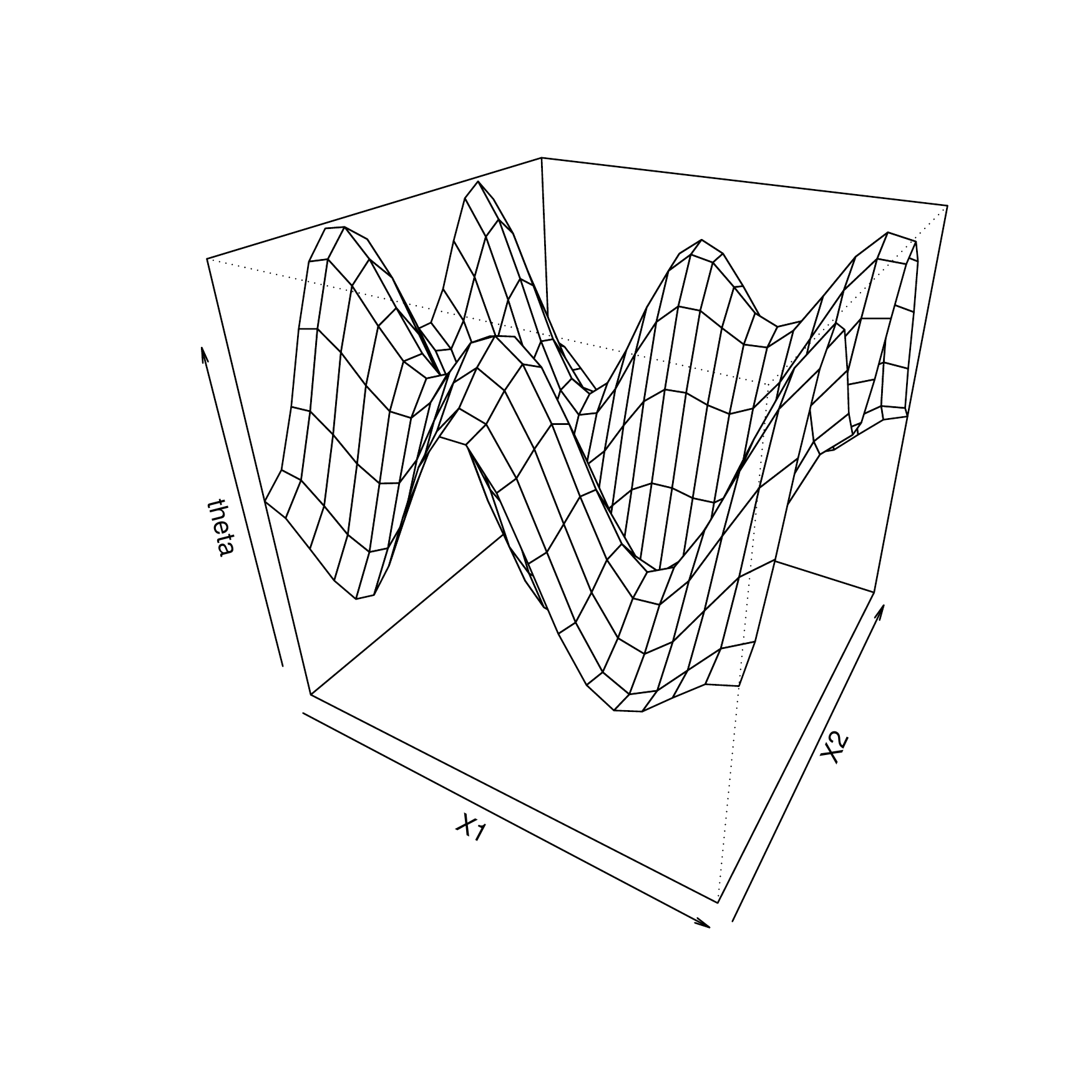}
\begin{center}{\it (b)}\end{center}
\end{minipage}
\begin{minipage}{0.3\textwidth}
\includegraphics[width=0.9\textwidth]{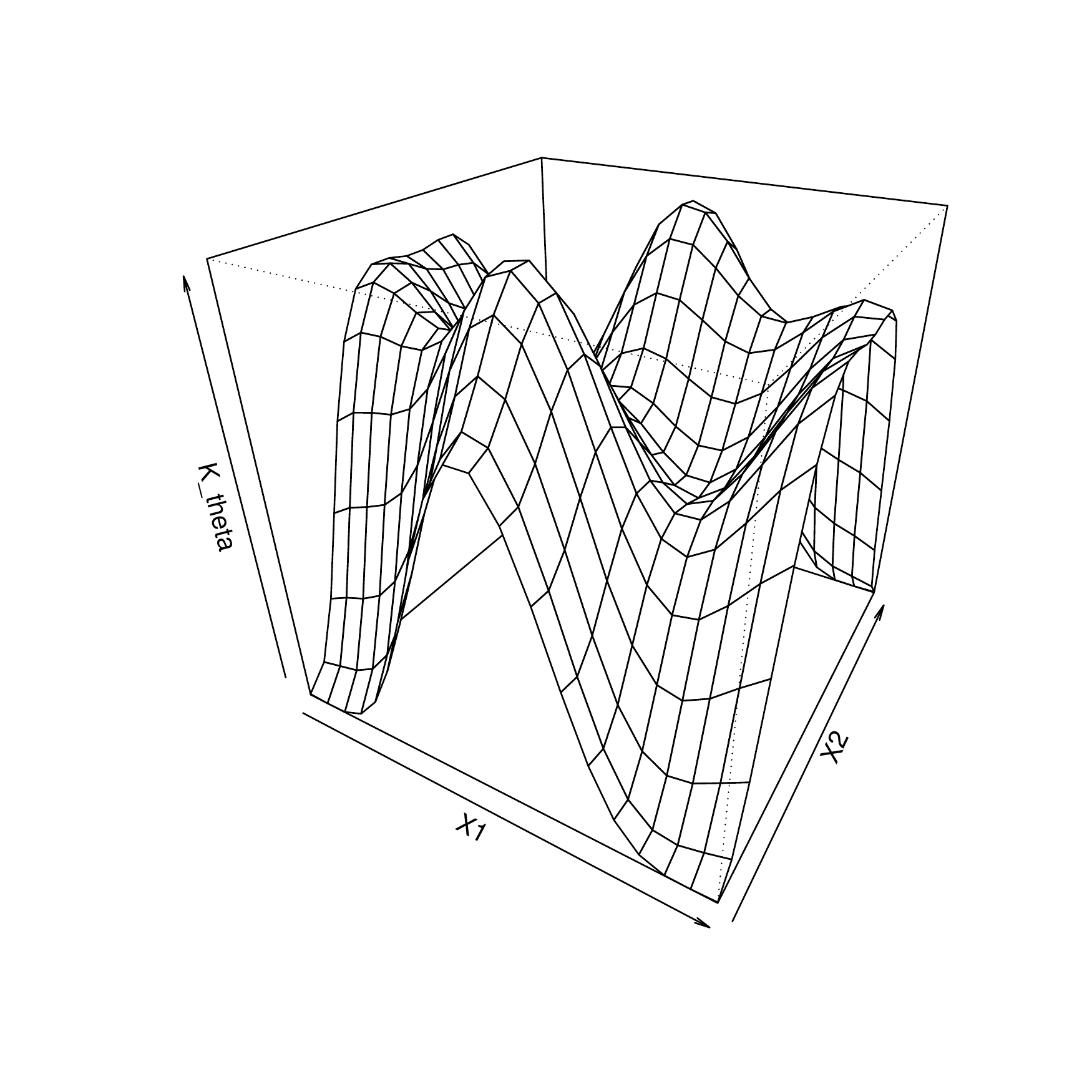}
\begin{center}{\it (c)}\end{center}
\end{minipage}
\caption{{\it Perspective plots of (a) the density function $g$, (b)
the indirect regression function $\theta$ and (c) the distorted
regression function $K\theta$.}}
\label{persp_plots}
\end{figure}

The distortion function $\psi$ is taken as the product of two
(normalized) Laplace density functions restricted to the interval
$[0,\,1]$, each with mean $1/2$ and scale $1/10$. For greater
transparency, the Fourier coefficients of the distortion function
$\psi$ are
\begin{equation*}
\Psi(\kvec) = \frac{\big( (-1)^{|k_1|} - \exp(-5) \big)
 \big( (-1)^{|k_2|} - \exp(-5) \big)}{(1 + 4\pi^2k_1^2/10^2)
 (1 + 4\pi^2k_2^2/10^2)(1 - \exp(-5))^2},
\quad \kvec = (k_1,\,k_2)^T \in \Zint^2.
\end{equation*}
This choice indeed satisfies Assumption \ref{assumpPsi} with
$b = 2$. When nonparametric smoothing is performed we work with the
radially symmetric spectral cutting kernel characterized by the
Fourier coefficient function $\Lambda(c_n\kvec) = \1[ \|c_n\kvec\|
\leq 1]$, $\kvec \in \Zint^2$, with smoothing parameter $c_n$ chosen
by minimizing the leave-one-out cross-validated estimate of the mean
squared prediction error (see, for example, H\"ardle and Marron,
1985). This choice is practical, simple to implement and performed
well in our study.

The indirect regression function is given by
\begin{align*}
\theta(x_1,\,x_2) &= 5 + \cos(2\pi x_1)
 + \frac32\cos(2\pi x_2) + \frac32\cos(4\pi x_1) \\
&\quad - 2 \cos(4\pi x_2) -2\cos\big(2\pi (x_1 + x_2)\big)
 -\frac12\cos\big(2\pi (x_1 - x_2)\big)
\end{align*}
for $(x_1,\,x_2)^T \in [0,\,1]^2$. This is easily seen to belong to
$\MM(s_0)$ for any $s_0 > 0$. Following the previous discussion, the
distorted regression $K\theta$ belongs to $\MM(s_0 + 2)$ for any
$s_0 > 0$. In the middle and right panels of Figure \ref{persp_plots}
we display the indirect regression function $\theta$ and the distorted
regression function $K\theta$.

\begin{table}
\centering
\begin{tabular}{|c| c c c c |}
\hline
\diagbox{$F$}{$n$} & $100$ & $200$ & $300$ & $500$ \\
\hline
Normal        & $0.048$ & $0.098$ & $0.072$ & $0.052$ \\
Laplace       & $0.209$ & $0.488$ & $0.713$ & $0.914$ \\
Skew-normal   & $0.136$ & $0.388$ & $0.577$ & $0.828$ \\
Student's $t$ & $0.211$ & $0.401$ & $0.586$ & $0.786$ \\
\hline
\end{tabular}
\vspace*{1em}
\caption{{\it Simulated power of the goodness-of-fit test \eqref{test}
for normally distributed errors at the $5\%$ level with sample
sizes $100$, $200$, $300$ and $500$ and with covariates having
non-trivial distribution characterized by the density function
$g$. The first row corresponds to N$(0,\,(1/2)^2)$ distributed
errors. The remaining rows display the powers of the test under the
fixed alternative error distributions: Laplace with scale parameter
$\sigma = 1/2$; centered, skew-normal with scale parameter $\sigma =
1$ and skew parameter $\alpha = 3$; Student's $t$ with $\nu = 6$
degrees of freedom.}}
\label{table_power_X_g}
\end{table}

\begin{table}
\centering
\begin{tabular}{|c| c c c c |}
\hline
\diagbox{$F$}{$n$} & $100$ & $200$ & $300$ & $500$ \\
\hline
Normal        & $0.039$ & $0.033$ & $0.032$ & $0.048$ \\
Laplace       & $0.318$ & $0.679$ & $0.872$ & $0.979$ \\
Skew-normal   & $0.226$ & $0.558$ & $0.740$ & $0.943$ \\
Student's $t$ & $0.270$ & $0.469$ & $0.640$ & $0.815$ \\
\hline
\end{tabular}
\vspace*{1em}
\caption{{\it Simulated power of the goodness-of-fit test \eqref{test}
for normally distributed errors at the $5\%$ level with sample sizes
$100$, $200$, $300$ and $500$ and with covariates independently,
uniformly distributed in $[0,\,1]^2$. The first row corresponds to
N$(0,\,(1/2)^2)$ distributed errors. The remaining rows display the
powers of the test under the fixed alternative error distributions:
Laplace with scale parameter $\sigma = 1/2$; centered, skew-normal
with scale parameter $\sigma = 1$ and skew parameter $\alpha = 3$;
Student's $t$ with $\nu = 6$ degrees of freedom.}}
\label{table_power_X_unif}
\end{table}

We considered four scenarios: normally distributed errors with
standard deviation $\sigma = 1/2$; Laplace distributed errors with
scale parameter $\sigma = 1/2$; centered, skew-normal errors with
scale parameter $\sigma = 1$ and skew parameter $\alpha = 3$
(standard deviation is $0.2265$); Student's $t$ distributed errors
with $\nu = 6$ degrees of freedom (standard deviation is
$1.2247$). The first scenario allows us to check the level of the
proposed test statistic $T_0$, and the other three scenarios allow for
observing the simulated powers of the proposed test. Here we work with
a $5\%$-level test, and the quantile $q_{0.05}$ is then $2.2414$.

We perform $1000$ simulation runs of samples of sizes $100$, $200$,
$300$ and $500$. Table \ref{table_power_X_g} displays the results for
the first example (when the covariates have the non-trivial
distribution characterized by the density function $g$) and Table
\ref{table_power_X_unif} displays the results for the second example
(when the covariates are independently, uniformly distributed in the
unit square $[0,\,1]^2$). Beginning with the first example, at the
sample size $100$ the test rejected the null hypothesis in $4.8\%$ of
the cases (near the desired $5\%$) but at the sample sizes $200$ and
$300$ the test respectively rejected the null hypothesis in $9.8\%$
and in $7.2\%$ of the cases, which are both above the desired $5\%$
nominal level. However, at the sample size $500$ the test rejected the
null hypothesis in $5.2\%$ of the cases, which is (again) near the
desired nominal level of $5\%$. We expect that this behavior is due to
the data-driven smoothing parameter selection. Interestingly, in the
second example the test is slightly conservative at all of the
simulated sample sizes (e.g.\ rejecting $3.2\%$ of the cases at sample
size $300$), but with sample size $500$ the test rejected the null
hypothesis in $4.8\%$ of the cases (near the nominal level of $5\%$),
which coincides with the first example.

Turning our attention now to the power of the test, in the first
example, we can see that the test performs well for moderate and
larger sample sizes. At the sample size $100$ the test respectively
rejected the alternative error distributions Laplace, skew-normal and
Student's $t$ in only $20.9\%$, $13.6\%$ and $21.1\%$ of the cases,
but at the sample size $500$ the test respectively rejected the
alternative distributions in $91.4\%$, $82.8\%$ and $78.6\%$ of the
cases. In the second example, we can see that the power of test
dramatically improves with smaller sample sizes (rejecting the
alternative distributions in $31.8\%$, $22.6\%$ and $27\%$ of the
cases at sample size $100$) with less improvement at larger sample
sizes (rejecting the alternative distributions in $97.9\%$, $94.3\%$
and $81.5\%$ of the cases at the sample size $500$).
In conclusion it appears that the proposed test statistic $T_0$ is an
effective tool for testing the goodness-of-fit of a desired error
distribution in indirect regression models.

\subsection{An application to image reconstruction}
\label{app}
Here we illustrate an application of the previous results using the
HeLa dataset investigated in Bissantz et al.\ (2009) and more recently
by Bissantz et al.\ (2016). This data composes an image of living HeLa
cells obtained using a standard confocal laser scanning microscope and
consists of intensity measurements (numbered values $0,\ldots,255$) on
$512 \times 512$ pixels giving a total of $262144$ observations, see
Figure \ref{HeLa_original}. As noted on page $41$ of Bissantz et al.\
(2009), these image data are (approximately) Poisson distributed. We
therefore apply the Anscombe transformation $Y \mapsto 2(Y +
3/8)^{1/2}$ to obtain approximately normally distributed data, and
then apply the test \eqref{test} to check the assumption of normally
distributed errors (at the $5\%$ level) from a reconstruction of this
image using the previously studied results. We use the computing
language {\em R} with the package {\em OpenImageR}, which allows for
reading the image data and conducting our analysis.

\begin{figure}
\includegraphics[width=0.4\textwidth]{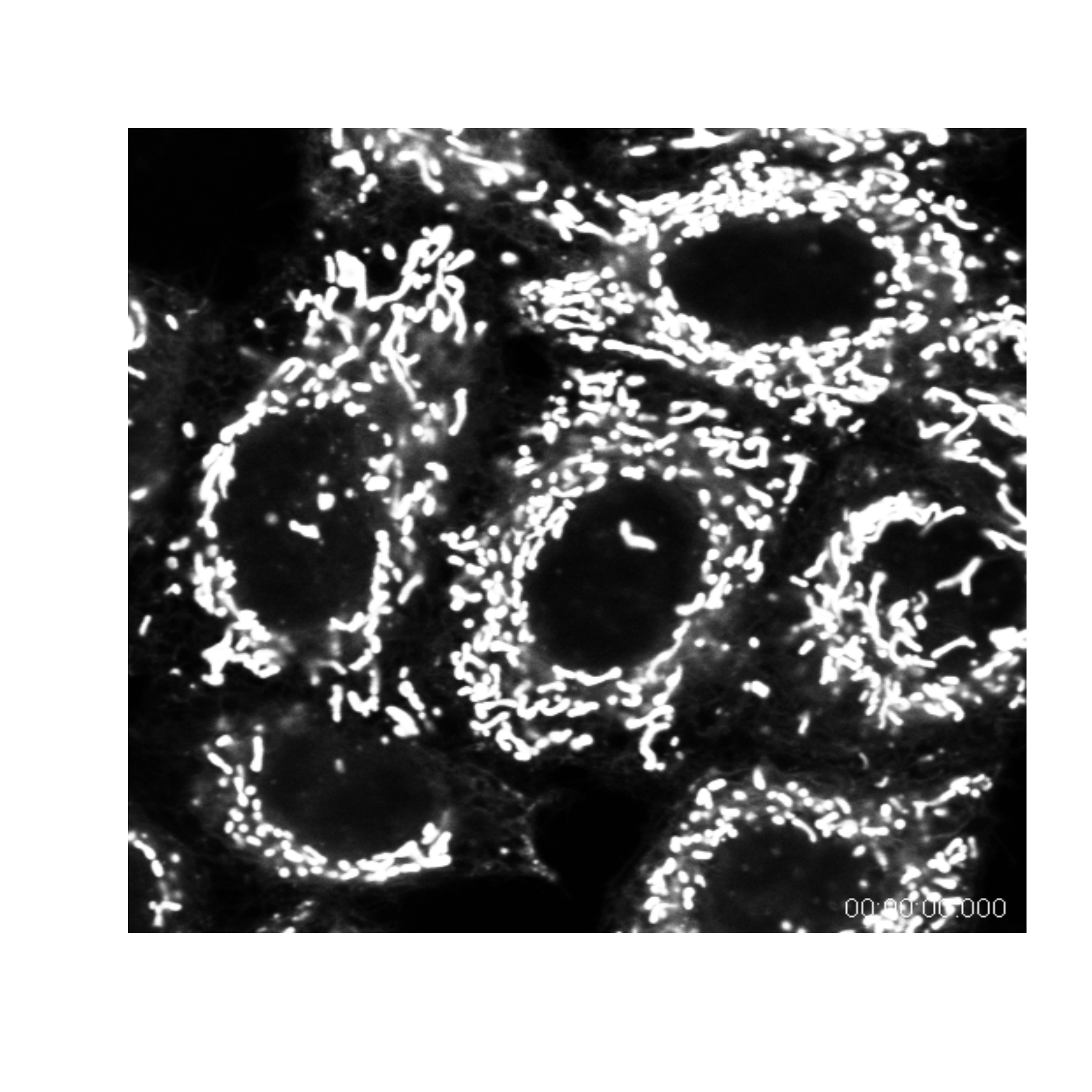}
\caption{{\it HeLa image data rendered in grayscale.}}
\label{HeLa_original}
\end{figure}

Since the total number of observations is quite large, we rather
illustrate the test for normal errors using two smaller sections of
the original HeLa image. To display the reconstructions of the smaller
images (for visual comparison with the original data) we apply the
inverse of the Anscombe transformation to the fitted values of each
regression. In both examples, the pixels are mapped to midpoints of
appropriate grids of the unit square $[0,\,1]^2$. The first image we
consider is $32 \times 32$ pixels composing $1024$ observations
and is displayed in Figure \ref{HeLa_small} alongside its reconstructed
version and a normal QQ-plot of the resulting standardized regression
residuals (see Section \ref{norm_test}). The second
image we consider is $64 \times 64$ pixels composing $4096$
observations and is displayed in Figure \ref{HeLa_medium} alongside
its reconstructed version and a normal QQ-plot of the resulting
standardized regression residuals. In both cases, as in Section
\ref{sim_study}, when nonparametric smoothing is applied the smoothing
parameter is chosen by minimizing the leave-one-out cross-validated
estimate of the mean squared prediction error.

\begin{figure}
\includegraphics[width=0.25\textwidth]{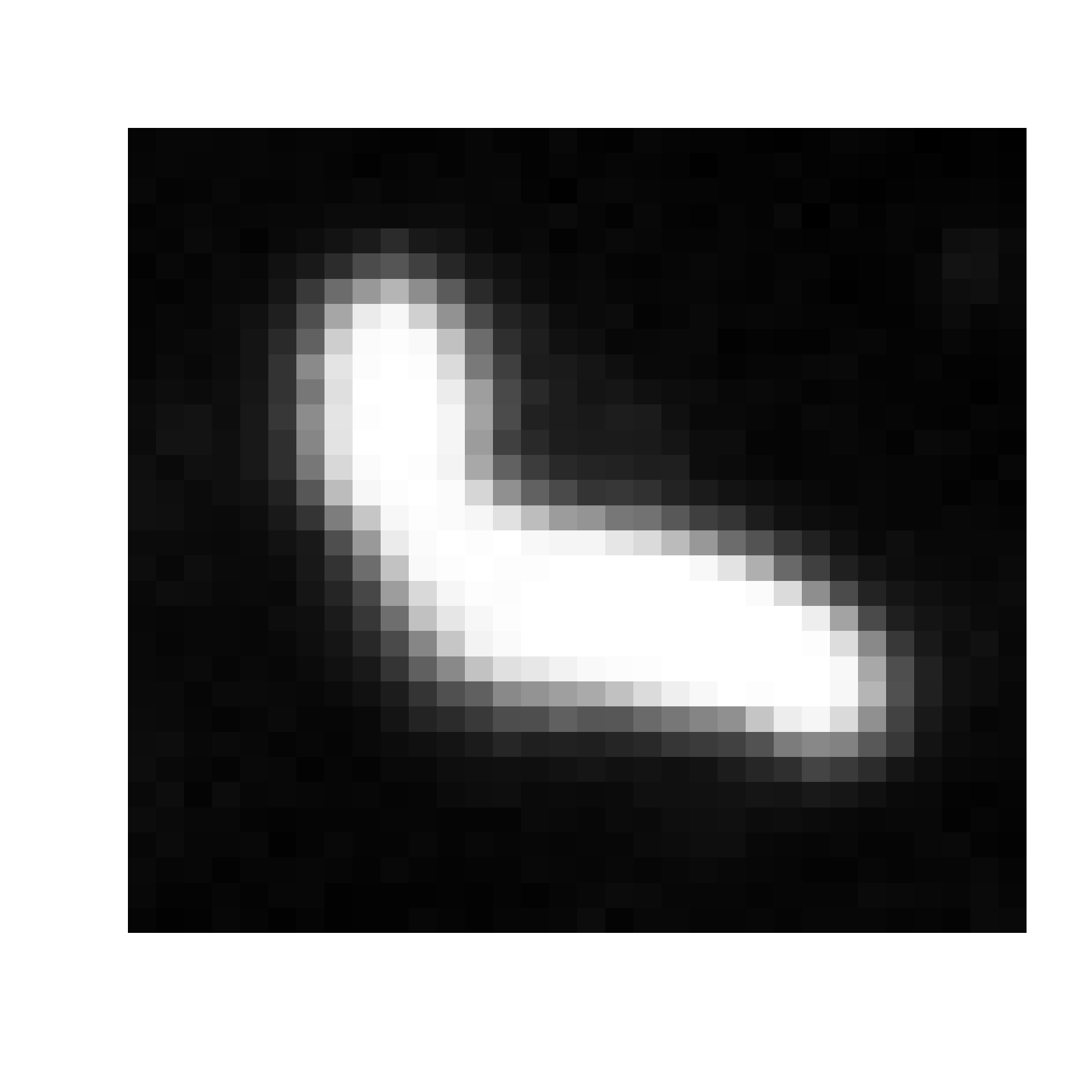}
\includegraphics[width=0.25\textwidth]{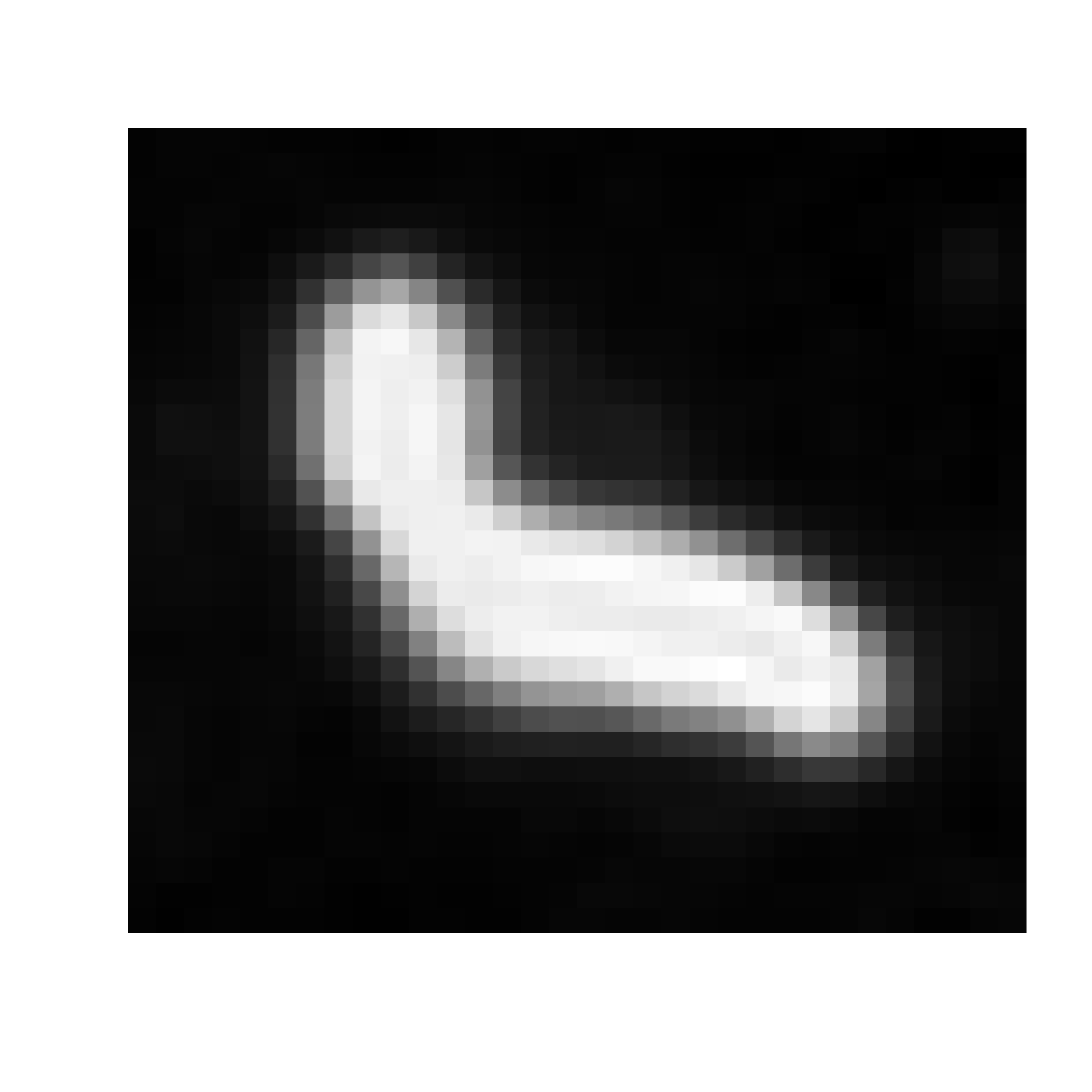}
\includegraphics[width=0.25\textwidth]{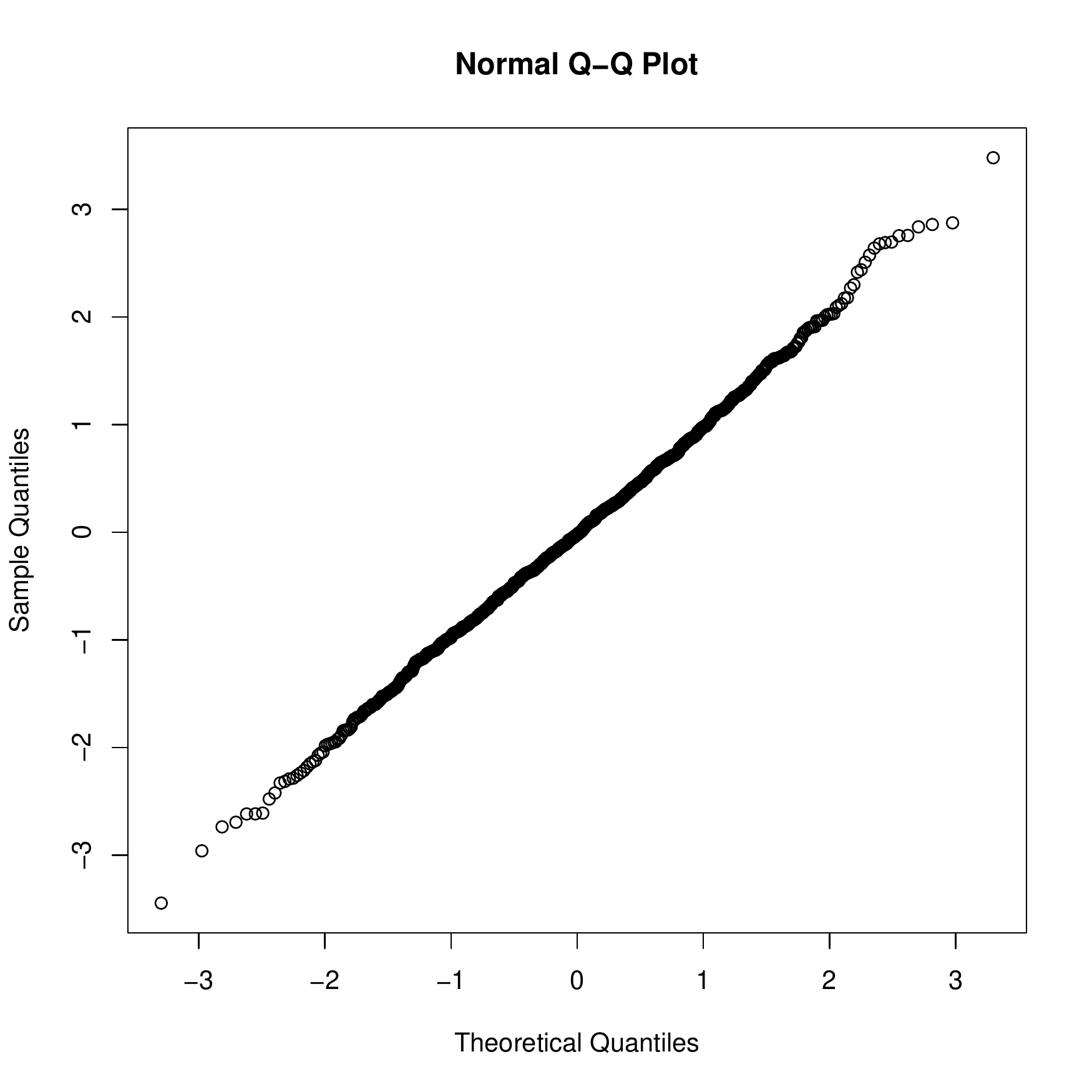}
\caption{{\it From left to right: $32 \times 32$ pixel section of the HeLa
image data rendered in grayscale, its reconstructed version
(grayscale), a normal QQ-plot of the resulting standardized regression
residuals.}}
\label{HeLa_small}
\end{figure}

\begin{figure}
\includegraphics[width=0.25\textwidth]{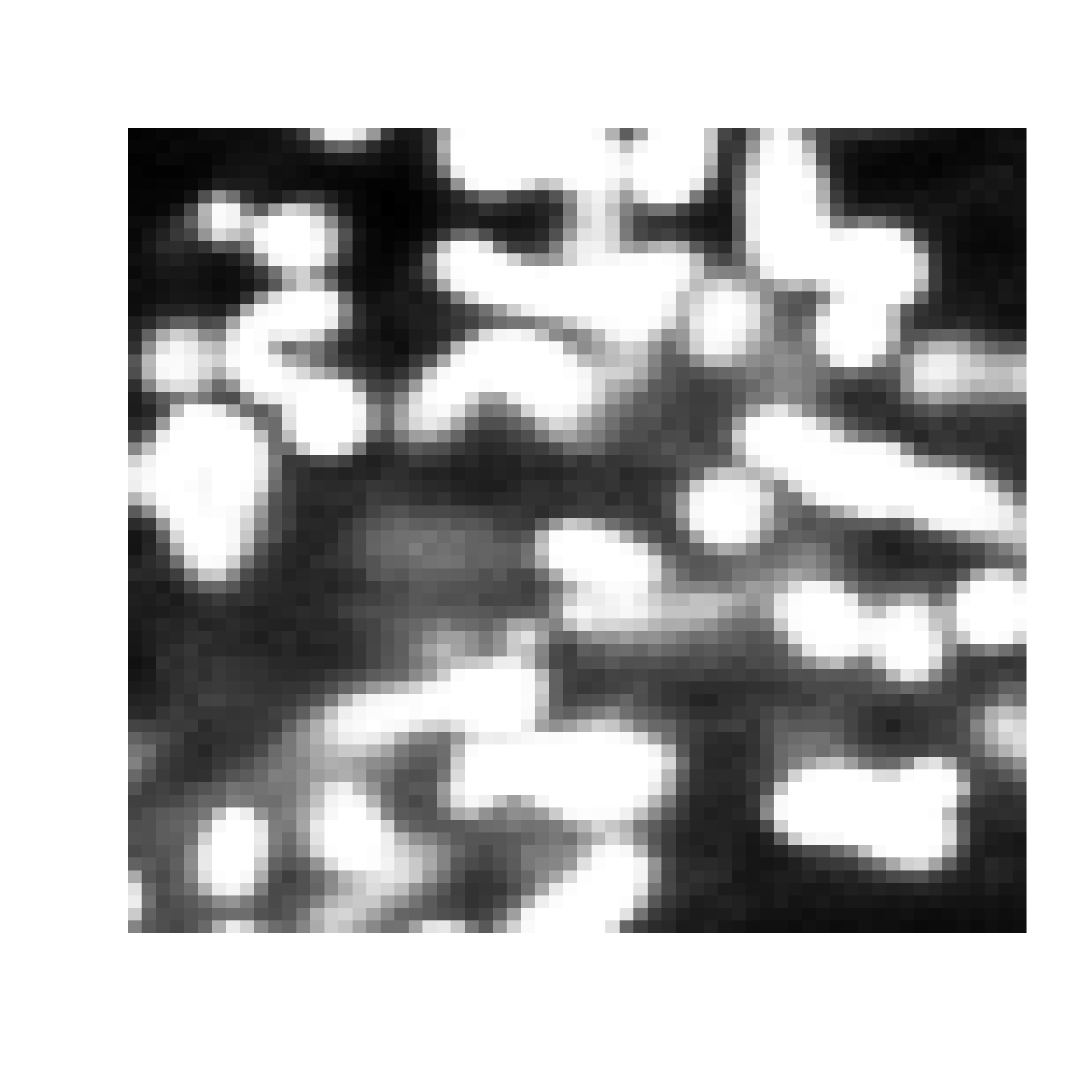}
\includegraphics[width=0.25\textwidth]{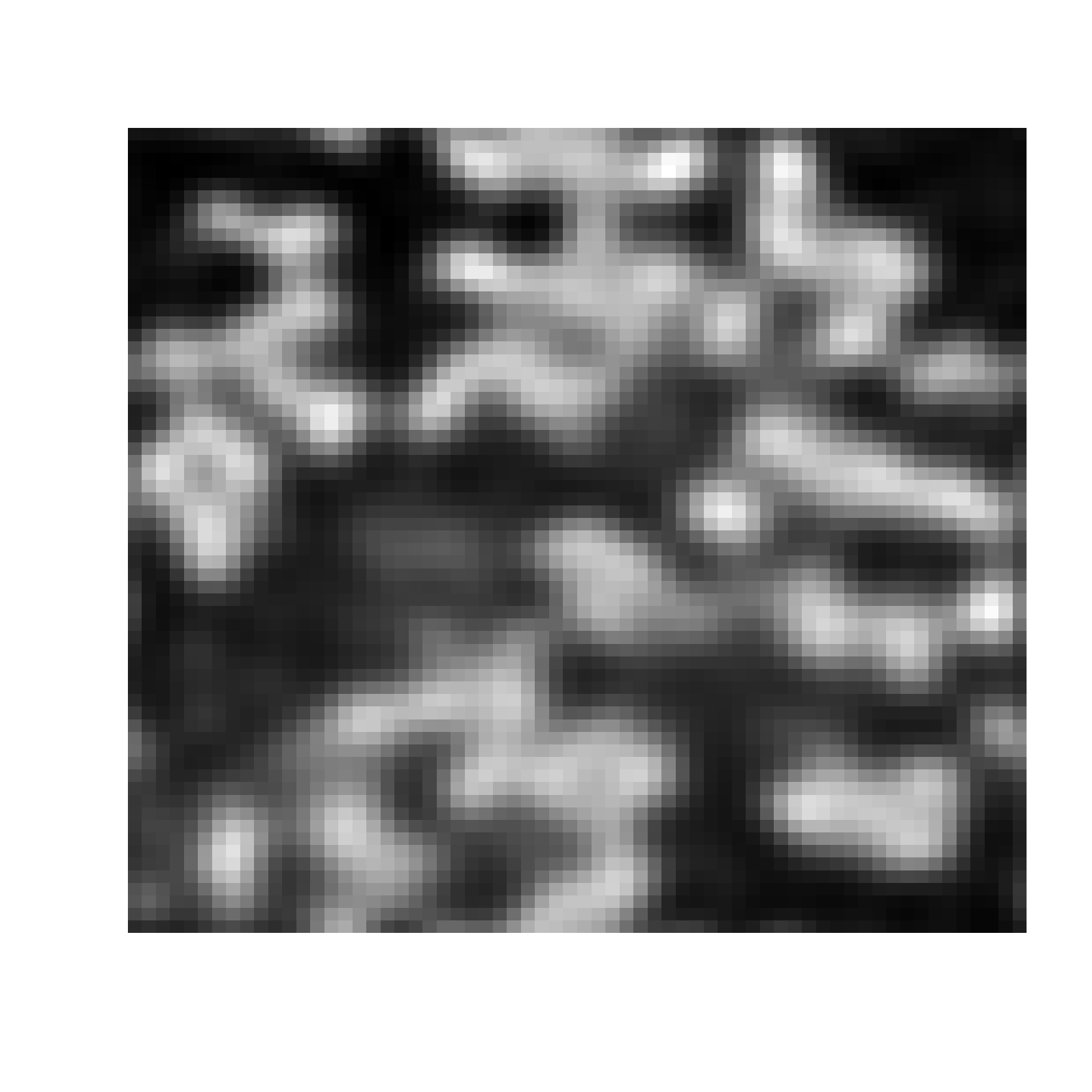}
\includegraphics[width=0.25\textwidth]{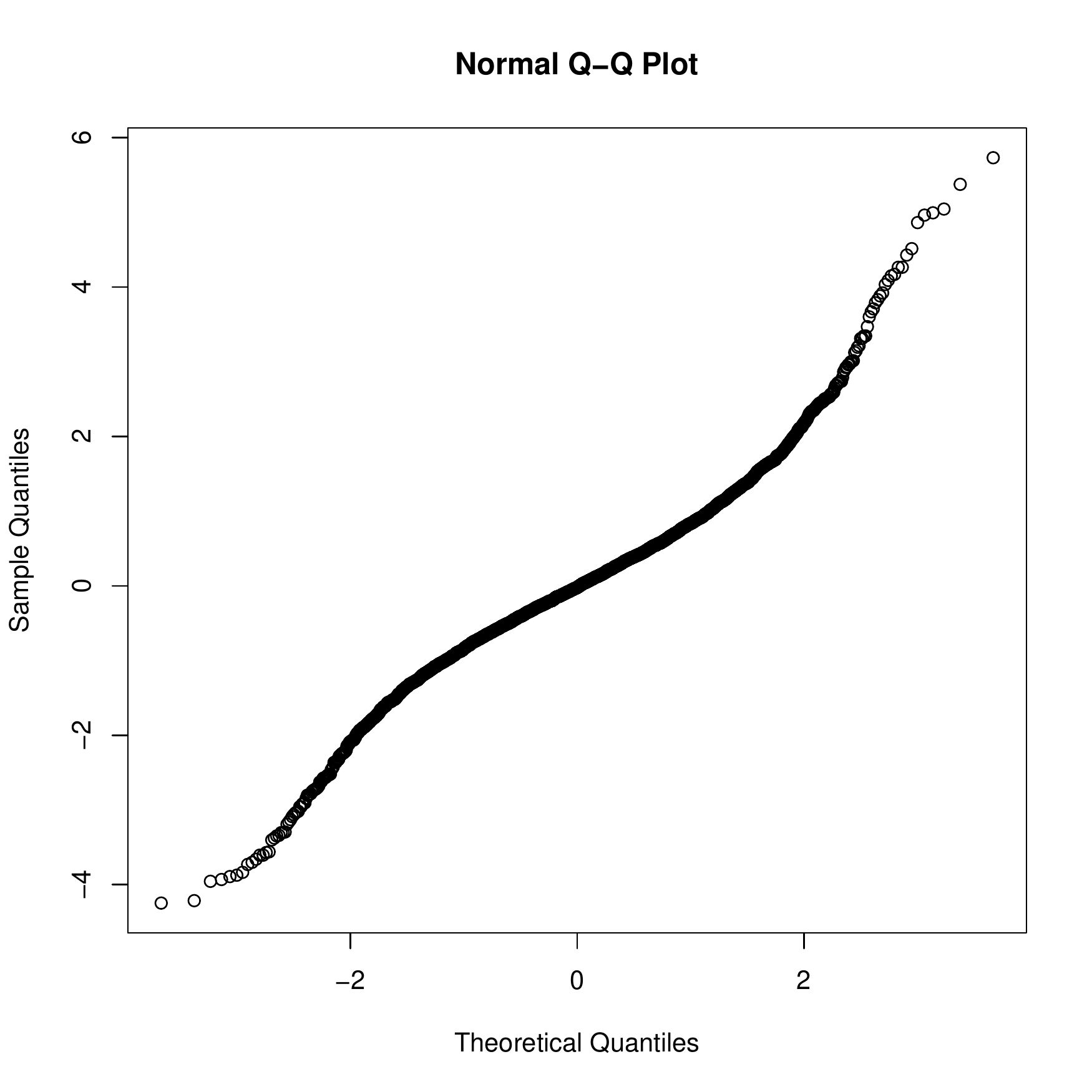}
\caption{{\it From left to right: $64 \times 64$ pixel section of the HeLa
image data rendered in grayscale, its reconstructed version
(grayscale), a normal QQ-plot of the resulting standardized regression
residuals.}}
\label{HeLa_medium}
\end{figure}

Beginning with the first and smaller image, the martingale transform
test statistic $T_0$ that assesses the goodness-of-fit of a normal
distribution has value $1.5141$, which is smaller than $2.2414$, and
the null hypothesis of normally distributed errors is not
rejected. Inspecting the QQ-plot of these standardized residuals it
appears that the assumption of normally distributed errors is
appropriate, which confirms our previous finding. In this case, we can
see the reconstruction very closely mirrors the original.

Turning now to the second and larger image, the value of the test
statistic is $39.8324$, which is much larger than $2.2414$, and we
reject the null hypothesis of normally distributed errors. The QQ-plot
of the standardized residuals now appears to contain systematic
deviation from normality, which confirms that the hypothesis of the
normally distributed errors is inappropriate. Here we can see the
reconstruction is now not as accurate as it was for the previous
case. In conclusion, we can see the approach of using the proposed
test statistic $T_0$ for assessing convenient forms of the error
distribution is useful.

\medskip
\medskip


{\bf Acknowledgements}
This work has been supported in part by the Collaborative Research
Center ``Statistical modeling of nonlinear dynamic processes'' (SFB
823, Project C1) of the German Research Foundation (DFG) and
in part by the Bundesministerium f\"ur Bildung und Forschung through
the project ``MED4D: Dynamic medical imaging: Modeling and analysis of
medical data for improved diagnosis, supervision and drug
development''. The authors would also like to thank Kathrin Bissantz
for providing us the HeLa cells image used in our data example.


\section{Appendix}
\label{appendix}

In this section we give the technical details supporting our
results. We have the following uniform convergence property for the
density estimator $\hg$.
\begin{lemma} \label{lem_hg}
Let the Fourier smoothing kernel $\Lambda$ be as in Assumption
\ref{assumpLambda}, and let
Assumption \ref{assump_g} hold with $s >0$. Then, for any
smoothing parameter sequence $\{c_n\}_{n \geq 1}$ satisfying
$(nc_n^m)^{-1}\log(n) \to 0$ as $c_n \to 0$ with $n \to \infty$,
\begin{equation} \label{lem1}
\sup_{\xvec \in \CC} \Big| \hg(\xvec) - g(\xvec) \Big|
 = O\big( c_n^s + (nc_n^m)^{-1/2}\log^{1/2}(n) \big),
 \quad\text{a.s.}
\end{equation}
\end{lemma}
\begin{proof}
Write
\begin{equation*}
E\big[ \hg(\xvec) \big] - g(\xvec) =
 \sum_{\kvec \in \Zint^m} \big\{ \Lambda(c_n\kvec) - 1 \big\}
 \phi_g(\kvec) e^{ i 2\pi \kvec \cdot \xvec },
 \quad \xvec \in \CC,
\end{equation*}
(and note that $|\Lambda(c_n\kvec) - 1| = 0$ whenever $\|\kvec\| \leq
c_n^{-1}$) to see that
\begin{equation*}
\sup_{\xvec \in \CC} \Big| E\big[ \hg(\xvec) \big] - g(\xvec) \Big| \leq
 2 c_n^s \sum_{\kvec \in \Zint^m} \|\kvec\|^s|\phi_g(\kvec)|
 = O\big( c_n^s \big).
\end{equation*}

Using the representation $L_\Lambda(\xvec) = \sum_{\kvec \in \Zint^m}
\Lambda(\kvec) e^{i 2\pi \kvec \cdot \xvec}$ and the fact that
$\{\Lambda(c_n\kvec)\}_{\kvec \in \Zint^m}$ are the Fourier
coefficients of the function $L_\Lambda(\cdot/c_n)/c_n^m$ we obtain
\begin{equation*}
\hg(\xvec) - E\big[ \hg(\xvec) \big] =
 \frac1{nc_n^m} \sum_{j = 1}^n \bigg\{
 L_\Lambda\bigg( \frac{\xvec - X_j}{c_n} \bigg) -
 E\bigg[ L_\Lambda\bigg( \frac{\xvec - X}{c_n} \bigg) \bigg] \bigg\},
 \quad \xvec \in \CC.
\end{equation*}
One calculates directly that
\begin{equation} \label{K_Lambda_variance}
\Var \bigg[ c_n^{-m} L_\Lambda \bigg( \frac{\xvec - X}{c_n} \bigg) \bigg]
 = O\big( c_n^{-m} \big),
 \quad \xvec \in \CC.
\end{equation}
In addition, $L_\Lambda$ is bounded and therefore
\begin{equation} \label{K_Lambda_bound}
c_n^{-m} \sup_{\xvec \in \CC} \bigg|
 L_\Lambda \bigg( \frac{\xvec - X_j}{c_n} \bigg) -
 E\bigg[ L_\Lambda \bigg( \frac{\xvec - X}{c_n} \bigg) \bigg] \bigg|
 = O\big( c_n^{-m} \big),
 \quad j = 1,\ldots,n.
\end{equation}

To continue, let $\{s_n\}_{n \geq 1}$ be a sequence of positive real
numbers satisfying $s_n = O(c_n^{m/2 + 1}) = o(1)$ and partition $\CC$
into parts $\CC_i$ with associated centers $\xvec_i$ ($i =
1,\ldots,O(s_n^{-m})$) such that $\max_{i = 1,\ldots,O(s_n^{-m})}
\sup_{\xvec \in \CC_i} \|\xvec - \xvec_i\| \leq s_n$. The assertion
\eqref{lem1} follows from the arguments above and by additionally
showing that $\max_{i = 1,\ldots,O(s_n^{-m})}|\hg(\xvec_i) -
E[\hg(\xvec_i)]| = O((nc_n^m)^{-1/2}\log^{1/2}(n))$ and $\max_{i =
1,\ldots,O(s_n^{-m})}\sup_{\xvec \in \CC_i}|\hg(\xvec) - E[\hg(\xvec)] -
\hg(\xvec_i) + E[\hg(\xvec_i)]| = O((nc_n^m)^{-1/2}\log^{1/2}(n))$,
almost surely.

Combining \eqref{K_Lambda_variance} and \eqref{K_Lambda_bound} with
Bernstein's inequality (see, for example, Section 2.2.2 of van der
Vaart and Wellner, 1996), one chooses a large enough positive constant
$C$ (through the choice of the quantity
$O((nc_n^m)^{-1/2}\log^{1/2}(n))$) such that
\begin{equation*}
P\bigg( \max_{i = 1,\ldots,O(s_n^{-m})}
 \big| \hg(\xvec_i) - E\big[ \hg(\xvec_i) \big] \big| >
 O\big( (nc_n^m)^{-1/2}\log^{1/2}(n) \big) \bigg)
 \leq O\big( s_n^{-m} n^{-C} \big)
\end{equation*}
is summable in $n$. Since $O( s_n^{-m} n^{-C} ) = O( (n^C c_n^{m^2/2 +
m})^{-1} )$, this occurs when $C > m/2 + 2$ and we have
\begin{equation} \label{hg_pointwise_consistent}
\max_{i = 1,\ldots,O(s_n^{-m})} \Big|
 \hg(\xvec_i) - E\big[ \hg(\xvec_i) \big] \Big|
 = O\big( (nc_n^m)^{-1/2}\log^{1/2}(n) \big),
 \quad\text{a.s.}
\end{equation}

We will now demonstrate that $\max_{\kvec \in \Zint^m}
|\hphig(\kvec) - \phi_g(\kvec)| = O(n^{-1/2}\log^{1/2}(n))$, almost
surely. Let $\kvec \in \Zint^m$ be arbitrary and write
\begin{equation*}
\hphig(\kvec) - \phi_g(\kvec) = \frac1n \sum_{j = 1}^n
 \Big\{ \exp(i2\pi\kvec\cdot X_j)
 - E\big[ \exp(i2\pi\kvec\cdot X) \big] \Big\},
\end{equation*}
where $X$ is a generic random variable with distribution characterized
by the density function $g$. The complex exponential functions are
bounded in absolute value by $1$, and it is easy to verify that
$\Var[\exp(i2\pi\kvec\cdot X)] \leq 1$. As above, use Bernstein's
inequality choosing a large enough positive constant $C$ (through
the choice of the quantity $O(n^{-1/2}\log^{1/2}(n))$) to find that
\begin{equation*}
P\bigg( \bigg| \frac1n \sum_{j = 1}^n
 \Big\{ \exp( i 2\pi \kvec \cdot X_j )
 - E\big[ \exp( i 2\pi \kvec \cdot X ) \big] \Big\} \bigg|
 > O\big(n^{-1/2}\log^{1/2}(n)\big) \bigg) \leq O\big(n^{-C}\big)
\end{equation*}
is summable in $n$. This occurs when $C > 1$, independent of
$\kvec$. It follows that $\max_{\kvec \in \Zint^m} |\hphig(\kvec) -
\phi_g(\kvec)| = O(n^{-1/2}\log^{1/2}(n))$, almost surely.

Further, let $\CC_i$ be arbitrary. For any $\xvec \in \CC_i$ it
follows that
\begin{equation} \label{hg_Lipschitz_remainder}
\hg(\xvec) - E\big[ \hg(\xvec) \big]
 - \hg(\xvec_i) + E\big[ \hg(\xvec_i) \big] =
 \sum_{\kvec \in \Zint^m} \Lambda(c_n\kvec)
 \big\{ \hphig(\kvec) - \phi_g(\kvec) \big\}
 \Big\{ e^{i 2\pi \kvec \cdot \xvec} - e^{i 2\pi \kvec \cdot \xvec_i} \Big\}.
\end{equation}
Now use Euler's formula to write
\begin{equation*}
\exp \big( -i 2\pi \kvec \cdot \xvec \big)
 = \cos \big( 2\pi \kvec \cdot \xvec \big)
 - i \sin \big( 2\pi \kvec \cdot \xvec \big),
\end{equation*}
and (using that sine and cosine are Lipschitz functions with constant
equal to one) derive the bound
\begin{equation} \label{Fourier_rotation_bound}
\Big| \exp\big( -i 2\pi \kvec \cdot \xvec \big)
 - \exp\big( -i 2\pi \kvec \cdot \xvec_i \big) \Big|
 \leq 2^{3/2} \pi \|\kvec\|\|\xvec - \xvec_i\|,
 \quad \xvec \in \CC_i.
\end{equation}
Combining \eqref{Fourier_rotation_bound} with
\eqref{hg_Lipschitz_remainder} there is a positive
constant $C > 0$ such that
\begin{align} \label{hg_Lipschitz}
&\max_{i = 1,\ldots,O(s_n^{-m})}
\sup_{\xvec \in \CC_i} \Big|
 \hg(\xvec) - E[\hg(\xvec)] - \hg(\xvec_i) + E[\hg(\xvec_i)]
 \Big| \\
&\leq C (c_n^{m + 1})^{-1}\max_{\kvec \in \Zint^m}
 \Big| \hphig(\kvec) - \phi_g(\kvec) \Big|
 \max_{i = 1,\ldots,O(s_n^{-m})} \sup_{\xvec \in \CC_i} \|\xvec - \xvec_i\|
 \bigg\{ c_n^m \sum_{\kvec \in \Zint^m} \|c_n\kvec\|
 \big| \Lambda(c_n\kvec) \big| \bigg\} \\ \nonumber
&= O\big( (c_n^{m + 1})^{-1} s_n n^{-1/2}\log^{1/2}(n) \big)
 = O\big( (nc_n^m)^{-1/2}\log^{1/2}(n) \big),
\end{align}
almost surely, since $c_n^m \sum_{\kvec \in \Zint^m}
\|c_n\kvec\||\Lambda(c_n\kvec)| \to
\int_{\RR^m}\,\|\uvec\||\Lambda(\uvec)|\,d\uvec < \infty$ by
Assumption \ref{assumpLambda}.
\end{proof}

With the results of Lemma \ref{lem_hg} we can state a result on the
asymptotic order of the estimated coefficients $\{\Rhat(\kvec)\}_{\kvec
\in \Zint^m}$, which now depend on the density estimator $\hg$.

\begin{lemma} \label{lem_Rhat}
Let $\theta \in \MM(s_0)$ for some $s_0 > 0$, and assume that the errors
$\ve_1,\ldots,\ve_n$ have a finite absolute moment of order $\kappa >
2$. Let the Fourier smoothing kernel $\Lambda$ be as in Assumption
\nolinebreak \ref{assumpLambda}, and let
Assumption \ref{assump_g} hold for some $s > 0$. Choose the sequence
of smoothing parameters $\{c_n\}_{n \geq 1}$ such that
$(nc_n^m)^{-1}\log(n) \to 0$ and $n^{-1/2}\log^{1/2}(n) = o(c_n^s)$
with $c_n \to 0$ as $n \to \infty$. Then
\begin{equation*}
\max_{\kvec \in \Zint^m} \Big| \Rhat(\kvec) - R(\kvec) \Big|
 = O\big( c_n^s + (nc_n^m)^{-1/2}\log^{1/2}(n) \big),
 \quad\text{a.s.}
\end{equation*}
\end{lemma}
\begin{proof}
Let $\kvec \in \Zint^m$ be arbitrary and write
\begin{equation*}
\Rhat(\kvec) - R(\kvec) = T_1(\kvec) + T_2(\kvec) + T_3(\kvec) + T_4(\kvec),
\end{equation*}
with
\begin{equation*}
T_1(\kvec) = \frac1n \sum_{j = 1}^n \bigg\{
 \frac{[K\theta](X_j)}{g(X_j)}e^{-i 2\pi \kvec \cdot X_j}
 - E\bigg[ \frac{[K\theta](X)}{g(X)}e^{-i 2\pi \kvec \cdot X} \bigg]
 \bigg\},
\end{equation*}
\begin{equation*}
T_2(\kvec) = \frac1n \sum_{j = 1}^n
 \frac{\ve_j}{g(X_j)}e^{-i 2\pi \kvec \cdot X_j},
\end{equation*}
\begin{equation*}
T_3(\kvec) = \frac1n \sum_{j = 1}^n \big[ K\theta \big](X_j)
 \Big\{ \hg^{-1}(X_j) - g^{-1}(X_j) \Big\}e^{-i 2\pi \kvec \cdot X_j}
\end{equation*}
and
\begin{equation*}
T_4(\kvec) = \frac1n \sum_{j = 1}^n \ve_j
 \Big\{ \hg^{-1}(X_j) - g^{-1}(X_j) \Big\}e^{-i 2\pi \kvec \cdot X_j}.
\end{equation*}
Since $\theta \in \MM(s_0)$ for some $s_0 > 0$, it follows that
$K\theta$ is bounded, and a standard argument shows that
$\max_{\kvec \in \Zint^m} |T_1(\kvec)|$ is of the order
$O(n^{-1/2}\log^{1/2}(n)) = o(c_n^s + (nc_n^m)^{-1/2}\log^{1/2}(n))$,
almost surely. Analogously, $\max_{\kvec \in \Zint^m}
|T_2(\kvec)|$ is of the order $o(c_n^s + (nc_n^m)^{-1/2}\log^{1/2}(n))$,
almost surely. From the result of Lemma \nolinebreak
\ref{lem_hg} we can see that $\max_{\kvec \in \Zint^m} |T_3(\kvec)|$
is of the order $O(c_n^s + (nc_n^m)^{-1/2}\log^{1/2}(n))$, almost
surely. Finally, with some technical effort one shows that
$\max_{\kvec \in \Zint^m}|T_4(\kvec)|$ is of the order $o(c_n^s +
(nc_n^m)^{-1/2}\log^{1/2}(n))$, almost surely.
\end{proof}

We are now ready to state the proof of Theorem
\ref{thm_htheta_random}.
\begin{proof}[{\sc Proof of Theorem \ref{thm_htheta_random}}]
Write
\begin{equation*}
\hKtheta(\xvec) - K\theta(\xvec) =
 \sum_{\kvec \in \Zint^m} \Lambda(c_n\kvec)
 \big\{ \Rhat(\kvec) - R(\kvec) \big\} e^{i 2\pi \kvec \cdot \xvec}
 + \sum_{\kvec \in \Zint^m} \big\{ \Lambda(c_n\kvec) - 1 \big\}
 R(\kvec) e^{i 2\pi \kvec \cdot \xvec},
 \quad \xvec \in \CC.
\end{equation*}
From Lemma \ref{lem_Rhat} and that $c_n^s =
O((nc_n^m)^{-1/2}\log^{1/2}(n))$ it follows for the first term
in the display above to have the order $O(c_n^{s - m}) =
O(c_n^{s_0 + b})$, almost surely, since $s = s_0 + b + m$. The second
term in the same display is not random and easily shown to be of the
order $O(c_n^{s_0 + b})$.

The second assertion follows from showing that $\hKtheta \in \MM(s_0
+ b)$ and combining this fact with the first assertion. The Fourier
coefficients of $\hKtheta$ are given by
\begin{equation*}
\Lambda(c_n\kvec)\Rhat(\kvec) = \Lambda(c_n\kvec) R(\kvec) +
 \Lambda(c_n\kvec)\big\{ \Rhat(\kvec) - R(\kvec) \big\},
 \quad \kvec \in \Zint^m,
\end{equation*}
and we can see that $|\Lambda(c_n\kvec)\Rhat(\kvec)|$ is bounded by
\begin{equation} \label{cTheta_bound}
|R(\kvec)| + \max_{\xi \in \Zint^m}
 \Big| \Rhat(\xi) - R(\xi) \Big||\Lambda(c_n\kvec)|.
\end{equation}
Since $\theta \in \MM(s_0)$ it follows that $\sum_{\kvec
\in \Zint^m} \|\kvec\|^{s_0 + b} |R(\kvec)| = \sum_{\kvec \in \Zint^m}
\|\kvec\|^b|\Psi(\kvec)|\|\kvec\|^{s_0}|\Theta(\kvec)| < \infty$ and
$K\theta \in \MM(s_0 + b)$. This means that we only
need to show that the series condition in the definition of
$\MM(s_0 + b)$ is satisfied for the second term in
\eqref{cTheta_bound}. This series condition results in the quantity
\begin{equation*}
\max_{\xi \in \Zint^m} \Big| \Rhat(\xi) - R(\xi) \Big|
 \sum_{\kvec \in \Zint^m} \|\kvec\|^{s_0 + b}|\Lambda(c_n\kvec)|.
\end{equation*}
We have already used that $\max_{\xi \in \Zint^m} |\Rhat(\xi) -
R(\xi)|$ is of the order $O(c_n^s)$, and by choice of $\Lambda$ the
series in the display above is of the order $O(c_n^{-s_0 - b -
m})$ as in the proof of Lemma \ref{lem_hg}. Combining these findings
we can see that the quantity in the display above is of the order
$O(c_n^{s - s_0 - b - m}) = O(1)$.
\end{proof}

The proof of Theorem \ref{thm_hF_aslin_H0f} follows from the above
results with an additional property of the distorted regression
estimator $\hKtheta$ and an approximation result for the difference
$\shat^2 - \sigma^2$.

\begin{proposition} \label{prop_htheta_expansion}
Choose the Fourier smoothing kernel $\Lambda$ to be radially
symmetric. Then the estimator $\hKtheta$ enjoys the property that
\begin{equation*}
\bigg| \frac1n \sum_{j = 1}^n
 \Big\{ \hKtheta(X_j) - K\theta(X_j) \Big\}
 - \frac1n \sum_{j = 1}^n \ve_j \bigg| = 0.
\end{equation*}
If the assumptions of Theorem \ref{thm_htheta_random} are satisfied
with $s_0 + b > 3m/2$, then the estimator $\shat$ enjoys the property
that
\begin{equation*}
\bigg| \shat^2 - \sigma^2
 - \frac1n \sum_{j = 1}^n \big\{ \ve_j^2 - \sigma^2 \big\} \bigg|
 = o\big( n^{-1/2} \big),
 \quad\text{a.s.}
\end{equation*}
\end{proposition}
\begin{proof}
Write
\begin{equation*}
\frac1n \sum_{j = 1}^n
 \Big\{ \hKtheta(X_j) - K\theta(X_j) \Big\}
 - \frac1n \sum_{j = 1}^n \ve_j
 = \frac1n \sum_{j = 1}^n Y_j \bigg\{
 \frac{ \sum_{k = 1}^n W_{c_n}( X_k - X_j ) }{
 \sum_{k = 1}^n W_{c_n}( X_j - X_k )}
 - 1 \bigg\}.
\end{equation*}
Since $\Lambda$ is radially symmetric, we have that $W_{c_n}(X_j -
X_k) = W_{c_n}(X_k - X_j)$ for every $1 \leq j,k
\leq n$. One combines this fact with the additional fact that $|Y_j|$
is finite with probability 1 for each $1 \leq j \leq n$ to finish the
proof of the first assertion.

To show the second assertion we need to use the results of Theorem
\ref{thm_htheta_random} as follows. Write
\begin{equation*}
\shat^2 - \sigma^2 - \frac1n \sum_{j = 1}^n
 \big\{ \ve_j^2 - \sigma^2 \big\}
 = R_{1,n} - 2R_{2,n},
\end{equation*}
with
\begin{equation*}
R_{1,n} = \frac1n \sum_{j = 1}^n
 \Big\{ \hKtheta(X_j) - K\theta(X_j) \Big\}^2
\end{equation*}
and
\begin{equation*}
R_{2,n} = \frac1n \sum_{j = 1}^n \ve_j
 \Big\{ \hKtheta(X_j) - K\theta(X_j) \Big\}.
\end{equation*}
Now combine the first result of Theorem \ref{thm_htheta_random} with
$s_0 + b > 3m/2$ to find that $|R_{n,1}| = o(n^{-1/2})$, almost
surely.

To continue, write
\begin{align*}
R_{2,n} &= \sum_{\kvec \in \Zint^m}
 \big\{ \Lambda(c_n\kvec) - 1 \big\} R(\kvec)
 \bigg\{ \frac1n \sum_{j = 1}^n \ve_j e^{i 2\pi \kvec \cdot X_j} \bigg\} \\
&\quad + \sum_{\kvec \in \Zint^m}
 \big\{ \Rhat(\kvec) - R(\kvec) \big\} \Lambda(c_n\kvec)
 \bigg\{ \frac1n \sum_{j = 1}^n \ve_j e^{i 2\pi \kvec \cdot X_j} \bigg\}
\end{align*}
to see that $|R_{2,n}|$ is bounded by
\begin{equation*}
\max_{\kvec \in \Zint^m} \bigg|
 \frac1n \sum_{j = 1}^n \ve_j e^{i 2\pi \kvec \cdot X_j} \bigg| \Bigg[
 \max_{\kvec \in \Zint^m}\Big| \Rhat(\kvec) - R(\kvec) \Big|
 \sum_{\kvec \in \Zint^m} \big| \Lambda(c_n\kvec) \big|
 + \sum_{\kvec \in \Zint^m} \big| \Lambda(c_n\kvec) - 1 \big|
 |R(\kvec)| \Bigg].
\end{equation*}
Analogously to the proof of Lemma \ref{lem_Rhat}, one treats
$\max_{\kvec \in \Zint^m}|n^{-1}\sum_{j = 1}^n \ve_j \exp(i 2\pi \kvec
\cdot X_j)|$ using a standard argument and finds this quantity is of
the order $O(n^{-1/2}\log^{1/2}(n))$, almost surely. For the
quantities inside the large brackets, one uses Lemma \ref{lem_Rhat}
and handles the series term as in the proof of Lemma \ref{lem_hg} to
show that the first term is of the order $O(c_n^{s - m}) = O(c_n^{s_0
+ b})$ (since $s = s_0 + b + m$) and the second term is easily shown
to be of the order $O(c_n^{s_0 + b})$ (see the proof of Lemma
\ref{lem_hg}). Therefore, $|R_{2,n}|$ is of the order $O(c_n^{s_0 + b}
n^{-1/2} \log^{1/2}(n)) = o(n^{-1/2})$, almost surely.
\end{proof}

Neumeyer and Van Keilegom (2010) consider estimation of the
distribution function of the standardized errors using a
residual-based empirical distribution function based on nonparametric
regression residuals obtained by local polynomial smoothing. These
authors obtain asymptotic negligibility of a modulus of continuity
relating their residual-based empirical distribution function to the
empirical distribution function of their regression model errors (see
Lemma A.3 in that article). We obtain a similar result for the
estimator $\hF$ (stated as a proposition below) using analogous
arguments to those of Neumeyer and Van Keilegom (2010). These
arguments have been omitted for brevity.

\begin{proposition} \label{prop_hF_modulus}
Let the assumptions of Theorem \ref{thm_htheta_random} be satisfied
with $s_0 + b > m$. Additionally, assume that $F_*$ admits a
bounded Lebesgue density $f_*$ that satisfies $\sup_{t \in \RR}
|tf_*(t)| < \infty$. Then under the null hypothesis $H_0$ in
\eqref{hyp_null}
\begin{equation*}
\sup_{t \in \RR} \bigg| \hF(t) - \frac1n \sum_{j = 1}^n
 F_* \bigg( t + \frac{\shat - \sigma}{\sigma}t
 + \frac{\hKtheta(X_j) - K\theta(X_j)}{\sigma} \bigg)
 - \frac1n \sum_{j = 1}^n \1[ Z_j \leq t ] + F_*(t) \bigg|
 = \opn.
\end{equation*}
\end{proposition}

We are now prepared the state the proof of Theorem
\ref{thm_hF_aslin_H0f}.
\begin{proof}[{\sc Proof of Theorem \ref{thm_hF_aslin_H0f}}]
We introduce the notation
\begin{equation*}
E_n(t) = \frac1n \sum_{j = 1}^n \bigg\{ \1[ Z_j \leq t ] - F_*(t)
 + f_*(t) \bigg( Z_j + t \frac{Z_j^2 - 1}{2} \bigg) \bigg\},
\quad t \in \RR,
\end{equation*}
and write
\begin{equation*}
\hF(t) - F_*(t) - E_n(t)
 = M_n(t) + H_n(t) + L_n(t)
 = D_n(t),
\quad t \in \RR,
\end{equation*}
where the remainder term $D_n(t)$ is equal to the sum of
\begin{equation*}
M_n(t) = \hF(t) - \frac1n \sum_{j = 1}^n
 F_* \bigg( t + \frac{\shat - \sigma}{\sigma}t
 + \frac{\hKtheta(X_j) - K\theta(X_j)}{\sigma} \bigg)
 - \frac1n \sum_{j = 1}^n \1[ Z_j \leq t ] + F_*(t),
\end{equation*}
\begin{align*}
H_n(t) &= \frac1n \sum_{j = 1}^n
 F_* \bigg( t + \frac{\shat - \sigma}{\sigma}t
 + \frac{\hKtheta(X_j) - K\theta(X_j)}{\sigma} \bigg)
 - F_*(t) \\
&\quad - f_*(t) \frac{\sigma^{-1}}{n} \sum_{j = 1}^n
 \Big\{\hKtheta(X_j) - K\theta(X_j)\Big\}
 - tf_*(t) \frac{\shat - \sigma}{\sigma},
\end{align*}
and
\begin{equation*}
L_n(t) = f_*(t) \bigg\{ \frac{\sigma^{-1}}{n} \sum_{j = 1}^n
 \Big\{\hKtheta(X_j) - K\theta(X_j)\Big\}
 - \frac1n \sum_{j = 1}^n Z_j \bigg\}
 + tf_*(t) \bigg\{ \frac{\shat - \sigma}{\sigma}
 - \frac1n \sum_{j = 1}^n \frac{Z_j^2 - 1}{2} \bigg\}.
\end{equation*}
From Proposition \ref{prop_hF_modulus} it follows that
$\sup_{t \in \RR} |M_n(t)| = \opn$. Proposition
\ref{prop_htheta_expansion} in combination with the bounding
conditions on $f_*$ imply that $\sup_{t \in \RR} |L_n(t)| = \opn$
(note that $Z_j = \ve_j/\sigma$, $j = 1,\ldots,n$).

To show that $\sup_{t \in \RR} |H_n(t)| = \opn$ and finish
the proof we need to rewrite $H_n(t) = H_{1,n}(t) + H_{2,n}(t) +
H_{3,n}(t)$, with $H_{1,n}(t)$ equal to
\begin{equation*}
\frac{\sigma^{-1}}{n} \sum_{j = 1}^n
 \Big\{\hKtheta(X_j) - K\theta(X_j)\Big\}
 \int_0^1\,\bigg\{ f_*\bigg( t + \frac{\shat - \sigma}{\sigma} t
 + \frac{\hKtheta(X_j) - K\theta(X_j)}{\sigma} s \bigg)
 - f_*\bigg( t + \frac{\shat - \sigma}{\sigma} t \bigg) \bigg\}\,ds,
\end{equation*}
\begin{equation*}
H_{2,n}(t) = \bigg\{ f_*\bigg( t + \frac{\shat - \sigma}{\sigma} t \bigg)
 - f_*(t) \bigg\} \frac{\sigma^{-1}}{n} \sum_{j = 1}^n
 \Big\{\hKtheta(X_j) - K\theta(X_j)\Big\}
\end{equation*}
and
\begin{equation*}
H_{3,n}(t) = \frac{\shat - \sigma}{\sigma}
 t\int_0^1\,\bigg\{ f_*\bigg( t + \frac{\shat - \sigma}{\sigma} ts \bigg)
 - f_*(t) \bigg\}\,ds.
\end{equation*}
The H\"older continuity of $f_*$ guarantees that
\begin{equation*}
\sup_{t \in \RR} \big| H_{1,n}(t) \big|
 \leq \frac{C_{f_*}}{(1 + \gamma)\sigma^{1 + \gamma}}
 \sup_{\xvec \in \CC} \Big| \hKtheta(\xvec) - K\theta(\xvec) \Big|^{1 + \gamma}
 = o\big(n^{-1/2}\big),
 \quad\text{a.s.,}
\end{equation*}
from Theorem \ref{thm_htheta_random} and that $3m/(2s_0 + 2b) < \gamma
\leq 1$, which is $\opn$ and writing $C_{f_*}$ for the H\"older
constant associated to $f_*$. Proposition \ref{prop_htheta_expansion}
and the uniform continuity of $f_*$ imply that $\sup_{t \in \RR}
|H_{2,n}(t)| = \opn$. Finally, Proposition \ref{prop_htheta_expansion}
and the finite fourth moment assumption guarantees that $\shat$ is a
root-$n$ consistent estimator of $\sigma$, and combining this fact
with the uniform continuity and boundedness of the function $t \mapsto
tf_*(t)$ implies that $\sup_{t \in \RR} |H_{3,n}(t)| = \opn$.
\end{proof}



\begin{thebibliography}{16}

\bibitem{A1995}
Adorf, H.M.\ (1995).
Hubble Space Telescope image restoration in its fourth year.
{\em Inverse Problems} {\bf 11}, 639-653.

\bibitem{BBDV2009}
Bertero, M., Boccacci, P., Desider\'a, G.\ and Vicidomini, G.\ (2009).
Image deblurring with Poisson data: from cells to galaxies.
{\em Inverse Problems} {\bf 25}, 123006.

\bibitem{BRS2006}
Bickel, P.J., Ritov, Y.\ and Stoker, T.M.\ (2006).
Tailor-made tests for goodness-of-fit to semiparametric hypotheses.
{\em Ann.\ Statist.} {\bf 34}, 721-741.

\bibitem{BBH2010}
Birke, M., Bissantz, N.\ and Holzmann, H.\ (2010).
Confidence bands for inverse regression models.
{\em Inverse Problems} {\bf 26}, 115020.

\bibitem{BH2008}
Bissantz, N.\ and Holzmann, H.\ (2008).
Statistical inference for inverse problems.
{\em Inverse Problems} {\bf 24}, 034009.

\bibitem{BCHM2009}
Bissantz, N., Claeskens, G., Holzmann, H.\ and Munk, A.\ (2009).
Testing for lack of fit in inverse regression - with applications to
biophotonic imaging.
{\em J.\ R.\ Statist.\ Soc.\ B} {\bf 71}, 25-48.

\bibitem{BDHB2016}
Bissantz, N., Dette, H., Hildebrandt, T.\ and Bissantz, K.\ (2016).
Smooth backfitting in additive inverse regression.
{\em Ann.\ Inst.\ Stat.\ Math.} {\bf 68}, 827-853.

\bibitem{CEKL2015}
Can, S.U., Einmahl, J.H.J., Khmaladze, E.V.\ and Laeven, R.J.A.\
(2015).
Asymptotically distribution-free goodness-of-fit testing for tail
copulas.
{\em Ann.\ Statist.} {\bf 43}, 878-902.

\bibitem{CT2002}
Cavalier, L.\ and Tsybakov, A.\ (2002).
Sharp adaptation for inverse problems with random noise.
{\em Probab.\ Theory Relat.\ Fields} {\bf 123}, 323-354.

\bibitem{C2008}
Cavalier, L.\ (2008).
Nonparametric statistical inverse problems.
{\em Inverse Problems} {\bf 24}, 034004.

\bibitem{D1955}
Darling, D.A.\ (1955).
The Cram\'er-Smirnov test in the parametric case.
{\em Ann.\ Math.\ Statist.} {\bf 26}, 1-20.

\bibitem{dBG_AM2000}
del Barrio, E., Gusta-Albertos, J.A.\ and Matran, C.\ (2000).
Contributions of empirical and quantile processes to the asymptotic
theory of goodness-of-fit tests.
{\em TEST} {\bf 9}, 1-96.

\bibitem{DH2009}
Dette, H.\ and Hetzler, B.\ (2009).
Khmaladze transformation of integrated variance processes with
applications to goodness-of-fit testing.
{\em Math.\ Methods Statist.} {\bf 18}, 97-116.


\bibitem{D1973}
Durbin, J.\ (1973).
{\em Weak convergence of the sample distribution function when the
parameters are estimated.}
{\em Ann.\ Statist.} {\bf 2}, 279-290.


\bibitem{E2010}
Evans, L.C.\ (2010).
{\em Partial Differential Equations.}
Graduate Studies in Mathematics.
American Mathematical Society.

\bibitem{F1991}
Fan, J.\ (1991).
On the optimal rates of convergence for nonparametric deconvolution
problems.
{\em Ann.\ Statist.} {\bf 19}, 1257-1272.

\bibitem{HM1985}
H\"ardle, W.\ and Marron, J.S.\ (1985).
Optimal bandwidth selection in nonparametric regression function
estimation.
{\em Ann.\ Statist.} {\bf 13}, 1465-1481.

\bibitem{HK2008}
Haywood, J.\ and Khmaladze, E.V.\ (2008).
On distribution-free goodness-of-fit testing of exponentiality.
{\em J.\ Econometrics} {\bf 143}, 5-18.


\bibitem{JKPR2004}
Johnstone, I.M., Kerkyacharian, G., Picard, D.\ and Raimondo, M.\
(2004).
Wavelet deconvolution in a periodic setting.
{\em J.R.\ Statist.\ Soc.\ B} {\bf 66}, 547-573.

\bibitem{JP2014}
Johnstone, I.M.\ and Paul, D.\ (2014).
Adaptation in some linear inverse problems.
{\em Stat} {\bf 3}, 187-199.

\bibitem{K1981}
Khmaladze, E.V.\ (1981).
Martingale approach in the theory of goodness-of-fit tests.
{\em Theory Probab.\ Appl.} {\bf 26}, 240-257.

\bibitem{K1988}
Khmaladze, E.V.\ (1988).
An innovation approach to goodness-of-fit tests in $\RR^m$.
{\em Ann.\ Statist.} {\bf 16}, 1503-1516.

\bibitem{KK2004}
Khmaladze, E.V.\ and Koul, H.L.\ (2004).
Martingale transforms goodness-of-fit tests in regression models.
{\em Ann.\ Statist.} {\bf 32}, 995-1034.

\bibitem{KK2009}
Khmaladze, E.V.\ and Koul, H.L.\ (2009).
Goodness-of-fit problem for errors in nonparametric regression:
distribution free approach.
{\em Ann.\ Statist.} {\bf 37}, 3165-3185.

\bibitem{KS2010}
Koul, H.L.\ and Song, W.\ (2010).
Conditional variance model checking.
{\em J.\ Statist.\ Plann.\ Inference} {\bf 140}, 1056-1072.

\bibitem{KSZ2018}
Koul, H.L., Song, W.\ and Zhu, X.\ (2018).
Goodness-of-fit testing of error distribution in linear measurement
error models.
{\em Ann.\ Statist.} {\bf 46}, 2479-2510.

\bibitem{KSU2014}
K\"uhn, T., Sickel, W.\ and Ullrich, T.\ (2014).
Approximation numbers of Sobolev embeddings - Sharp constants and
tractability.
{\em J.\ Complexity} {\bf 30}, 95-116.

\bibitem{MR1996}
Mair, B.A.\ and Ruymgaart, F.H.\ (1996).
Statistical inverse estimation in Hilbert scales.
{\em SIAM J.\ Appl.\ Math.} {\bf 56}, 1424-1444.

\bibitem{MM2014}
Marteau, C.\ and Math\'e, P.\ (2014).
General regularization schemes for signal detection in inverse
problems.
{\em Math.\ Methods Statist.} {\bf 23}, 176-200.

\bibitem{MM1997}
Marzec, L.\ and Marzec, P.\ (1997).
Generalized martingale-residual processes for goodness-of-fit
inference in Cox's type regression models.
{\em Ann.\ Statist.} {\bf 25}, 683-714.


\bibitem{MSW2012}
M\"uller, U.U., Schick, A.\ and Wefelmeyer, W.\ (2012).
Estimating the error distribution function in semiparametric additive
regression models.
{\em J.\ Statist.\ Plann.\ Inference} {\bf 142}, 552-566.

\bibitem{NVK2010}
Neumeyer, N.\ and Van Keilegom, I.\ (2010).
Estimating the error distribution function in nonparametric multiple
regression with applications to model testing.
{\em J.\ Multivariate Anal.} {\bf 101}, 1067-1078.

\bibitem{PR1999}
Politis, D.N.\ and Romano, J.P.\ (1999).
Multivariate density estimation with general flat-top kernel of
infinite order.
{\em J.\ Multivariate Anal.} {\bf 68}, 1-25.

\bibitem{PBD2015}
Proksch, K., Bissantz, N.\ and Dette, H.\ (2015).
Confidence bands for multivariate and time dependent inverse
regression models.
{\em Bernoulli} {\bf 21}, 144-175.

\bibitem{SW1986}
Shorack, G.R.\ and Wellner, J.A.\ (1986).
{\em Empirical processes with Applications to Statistics.}
Wiley, New York.

\bibitem{STZ1998}
Stute, W., Thies, S.\ and Zhu, L.\ (1998)
Model checks for regression: an innovation process approach.
{\em Ann.\ Statist.} {\bf 26}, 1916-1934.

\bibitem{S1972}
Sukhatma, S.\ (1972).
Fredholm determinant of a positive kernel of a special type and its
applications.
{\em Ann.\ Math.\ Statist.} {\bf 43}, 1914-1926.

\bibitem{vdVW1996}
van der Vaart, A.W.\ and Wellner, J.A.\ (1996).
{\em Weak convergence and empirical processes. With applications to
statistics.}
Springer Series in Statistics. Springer-Verlag, New York.


\end{thebibliography}
\end{document}